\documentclass[12pt]{amsart}
\usepackage[margin=1in]{geometry}
\usepackage{amssymb}
\usepackage{stmaryrd}
\usepackage{mathrsfs}
\usepackage{graphicx, overpic}
\usepackage[all,cmtip]{xy}
\usepackage{url}
\usepackage{tensor}

\usepackage{todonotes}

\usepackage{hyperref}

\newcommand{\pder}[2]{\ensuremath{\frac{\partial #1}{\partial #2}}}
\newcommand{\so}{\ensuremath{\mathfrak{so}}}
\newcommand{\R}{\ensuremath{\mathbb{R}}}
\newcommand{\dx}{\ensuremath{\textrm{d}x}}
\newcommand{\dz}{\ensuremath{\textrm{d}z}}
\def\contract{\makebox[1.2em][c]{\mbox{\rule{.6em}
{.01truein}\rule{.01truein}{.6em}}}}

\newtheorem{thm}{Theorem}[section]
\newtheorem{prop}[thm]{Proposition}

\newtheorem{cor}[thm]{Corollary}
\newtheorem{defn}[thm]{Definition}

\DeclareMathOperator{\SDiff}{SDiff}
\DeclareMathOperator{\Diff}{Diff}
\DeclareMathOperator{\Jet}{Jet}
\DeclareMathOperator{\SO}{SO}
\DeclareMathOperator{\SL}{SL}
\DeclareMathOperator{\tr}{tr}
\DeclareMathOperator{\iso}{iso}
\DeclareMathOperator{\kernel}{kernel}
\DeclareMathOperator{\bag}{bag}
\DeclareMathOperator{\lie}{\mathcal{L}}
\DeclareMathOperator{\Ad}{Ad}
\DeclareMathOperator{\ad}{ad}

\graphicspath{{./images/}}
\title[Weak dual pairs and jetlet methods for ideal fluids]%
{Weak dual pairs and jetlet methods for ideal incompressible fluid models in $n \geq 2$ dimensions}

\usepackage[foot]{amsaddr}
\author{C.~J. Cotter$^1$, J. Eldering$^1$, D.~D. Holm$^1$, H.~O. Jacobs$^{1,*}$,}
\address{$^1$ Department of Mathematics, Imperial College London, London SW7 2AZ, UK}
\author{D.~M. Meier$^2$}
\address{$^2$ Department of Mathematics, Brunel University London, Uxbridge UB8 3PH, UK}
\address{$^*$ Corresponding author (h.jacobs@ic.ac.uk)}
\begin{document}

\begin{abstract}
  We review the role of dual pairs in mechanics and use them
  to derive particle-like solutions to regularized incompressible
  fluid systems.
  In our case we have a dual pair resulting from the action of
  diffeomorphisms on point particles (essentially by moving the points).
  We then augment our dual pair by considering the action of diffeomorphisms
  on Taylor series, also known as \emph{jets}.
  The augmented \emph{weak} dual pairs induce a hierarchy of particle-like solutions
  and conservation laws
  with particles carrying a copy of a jet group.
  We call these augmented particles \emph{jetlets}.
  The jet groups serve as finite-dimensional models of the diffeomorphism
  group itself, and so the jetlet particles serve as a finite-dimensional
  model of the self-similarity exhibited by ideal
  incompressible fluids.
  The conservation law associated to jetlet solutions is 
  shown to be a shadow of Kelvin's circulation theorem.
  Finally, we study the dynamics of infinite time particle mergers.
  We prove that two merging particles at the zeroth level in the hierarchy
  yield dynamics which asymptotically approach that of a single particle in the first level in the hierarchy.
  This merging behavior is then verified numerically
  as well as the exchange of angular momentum which must
  occur during a near collision of two particles.
  The resulting particle-like solutions suggest a new class
  of meshless methods which work in dimensions $n \geq 2$
  and which exhibit a shadow of Kelvin's circulation theorem.
  More broadly, this provides one of the first finite-dimensional models of
  self-similarity in ideal fluids.
\end{abstract}

\maketitle

\section{Introduction}
Arnold's geometric insight in \cite{Arnold1966}  has forever changed the way mathematicians look at ideal fluid dynamics.  According to \cite{Arnold1966}, ideal incompressible fluid motion on an
orientable Riemannian manifold $M$ is equivalent to geodesic motion on the Lie group of volume preserving diffeomorphisms $\SDiff(M)$ (i.e., the smooth invertible volume preserving maps of $M$ into itself, with smooth inverses).
The Riemannian metric is simply the fluid kinetic energy, which is the $L^2$-norm of the fluid's velocity field.
This characterization of ideal fluid flow allowed Poisson geometers and geometric mechanicians to provide a new perspective on ideal fluids and other PDEs with hydrodynamics background \cite{EbinMarsden1970,MarsdenWeinstein1983,Zeitlin1991,ArnoldKhesin1998,HolmMarsdenRatiu1998,FoiasHolmTiti2001}.
As will be shown, such a perspective is particularly fruitful for the purpose of reducing the infinite-dimensional fluid system
to a finite-dimensional ordinary differential equation (e.g.\ point vortex solutions).

The dimension reduction that will be performed in this paper and the method which it proposes
bears much semblance to the point vortex method \cite{Chorin1973}.
In particular, \cite{MarsdenWeinstein1983} illustrated how the point vortex solutions and the conservation of circulation could be derived via a pair of Poisson maps known as a dual pair.
One of these Poisson maps was the embedding map from a finite-dimensional manifold into the infinite-dimensional space dual to the divergence free vector fields.
In other words, the space of point vortices is merely an invariant manifold of ideal fluid motion.\footnote{This is not quite correct, as one must ignore the infinite self-energy terms of point vortices and then extend the space of admissible solutions to non-smooth velocity fields.  However, modulo this physically motivated caveat the statement holds.}
It was later realized that the use of dual pairs related to fluid applications was problematic and needed to be relaxed.
This led to the notion of weak dual pairs \cite{GayBalmazVizman2012}.
In this paper we will derive a hierarchy of different weak dual pairs, in order to obtain a class of finite-dimensionally parametrized solutions,
each of which comes with a conservation law that shadows Kelvin's circulation theorem.
Just as one can consider the point vortices to be the atoms of the point vortex method,
the atoms of these new solutions are particle-like objects which we call \emph{jetlets} (or $k$-\emph{jetlets} if we wish to emphasize that we are considering the $k$-th level of the hierarchy).

It is notable that the zeroth level of the hierarchy is a classical particle-like solution which appears in many geodesic
systems on diffeomorphism groups \cite{CamassaHolm1993,JoshiMiller2000,FringerHolm2001,MumfordMichor2013}.
It is also notable that, unlike a point vortex, a jetlet is also well-defined in dimensions greater than two.
\subsection{Main contributions}
 In this article we derive a hierarchy of particle-like
solutions for a regularized model of ideal fluids described in  \cite{MumfordMichor2013} as motion on the group of volume preserving
diffeomorphisms of $\R^n$, denoted $\SDiff(\R^n)$.
Each level in the hierarchy consists of particles
with internal group variables, parametrized by
a finite-dimensional model of a diffeomorphism group 
known as a jet group \cite[Chapter 4]{KMS99}.
The jet groups lie at the foundations of certain representation theories of diffeomorphism groups (see \cite[Appendix 2]{VershikGelfandGraev1975} or \cite{Kirillov1981}).
Hence, it seems natural to invoke these foundations in the context of fluids, where the configuration manifold is a diffeomorphism group.
The particle-like solutions, which ``carry'' jet groups, paint an intuitive picture of a large scale diffeomorphism that advects particles,
each of which carries its own partial description of a local deformation of the fluid in a small region around it.
At each higher level in the hierarchy, the description of the deformation becomes more detailed.
Thus the jets of diffeomorphisms possess a natural sense of the ``self-similarity'' which is present in the diffeomorphism group itself.
 Models such as this are crucial to our understanding of fluids,
  both from a numerical perspective and from the perspective
  of fundamental mathematics.

 Specifically, we will accomplish the following:
\begin{enumerate}
  \item We provide the first explicit Hamiltonian description of the full jet hierarchy of the particle-like solutions discovered in 
  	\cite{JacobsRatiuDesbrun2013,CotterHolmJacobsMeier2014}.
  \item We compute a nested sequence of conserved quantities at 
    each level in the hierarchy.
    Each of these conserved quantities will be related to the conservation
    of circulation.
  \item We numerically compute some particle-like solutions at the zeroth and first levels in the hierarchy.  We will observe a form of cascade in which interactions of solutions at the $k$-th level tend asymptotically in time toward solutions at the $(k+1)$-th level.
\end{enumerate}

\subsection{Outline of the paper}
After the introduction, in Section~\ref{sec:approach}, we outline our strategy: The first goal is to show that jetlet models admit a weak dual pair at each level in the hierarchy. Once this has been done, Theorem~\ref{thm:dual_pairs} (proven in the appendix) does the rest. Namely, the jetlets satisfy canonical Hamiltonian equations, and the right momentum map is conserved by the flow.

Section~\ref{sec:rigid_body} contains a brief discussion of the standard dual pair for the example of the  rigid body.

In Section~\ref{sec:reg_fluids} the main development of the paper starts. We first discuss briefly the Lie-Poisson approach to ideal fluids and Euler's equation. This is to set the scene for what follows. In Section~\ref{sec:MME} we discuss the Mumford-Michor model~\cite{MumfordMichor2013} and mention its standard dual pair (in parallel with the rigid body, the two legs of the dual pair correspond to the cotangent lift momentum maps for right and left actions, respectively).
We also recall that the conservation of the right momentum map $J_R$ here is equivalent to Kelvin's circulation theorem. Again, the dual pair viewpoint is extremely efficient. Once one has realized that there is a dual pair, one knows (from right-invariance) that the left momentum map $J_L$ maps Hamilton's equations on $T^*\SDiff$ to (reduced) Lie-Poisson form, and at the same time that $J_R$ is conserved.

In Section~\ref{sec:zeroth_order} we introduce zeroth order jetlets, also called landmarks, and we discuss their dual pair. This is done in parallel with later sections. The left momentum map is the usual one, and the right momentum map is in fact trivial. So the `dual pair' is somewhat unnatural here, but it helps present already the kind of thinking that will be employed in the later parts of the paper. 
%

In Section~\ref{sec:first_order} we consider first order jetlets. Everything here goes in parallel, except that one has to be more careful when introducing the relevant right and left actions. Once one has the definitions in place, it is not difficult to recognize that the cotangent lifts for the right and left actions lead to, now, a \emph{weak} dual pair of momentum maps.
We have to add `weak' here since the transitivity of the left action on level sets of the right momentum map is lost, but a weak dual pair is retained, essentially because the group actions still commute, see Definition~\ref{def:weak_dual_pair}.  Proposition~\ref{prop:1-solutions} collects the results that follow immediately as a consequence of the weak dual pair.

Throughout the paper we have taken care to explain the intuition behind the more abstract concepts. The discussion just after Proposition~\ref{prop:1-solutions} is an example of this. There, we discuss the relationship between the jetlet solutions on the diffeomorphism group and the resulting trajectory on the space of Taylor jets.

In Section~\ref{sec:higher_order} we describe jetlets at general levels. Again, once the relevant spaces and actions have been introduced (which is now quite an intricate endeavor), one recognizes that the momentum maps are a weak dual pair by the usual arguments. We give the form of the momentum maps explicitly just after Proposition~\ref{HL_action_prop}. Again, the weak dual pair leads to analogous conclusions about dynamics. Namely, it is canonically Hamiltonian, with $J_R$  conserved; $J_L$ maps the dynamics to the Lie-Poisson dynamics on the one-form densities.

In Section~\ref{sec:Kelvin} we discuss in more detail Kelvin's circulation theorem. More precisely, we discuss the relationship between the `standard' circulation theorem for the fluid and the conserved momentum maps $J_R$  at the various levels of jetlets. The main result is represented schematically and proven in detail (all that is required, in essence, is the standard formula for cotangent lift momentum maps). Intuitively speaking, what we show is that the conserved jetlet momentum maps are `shadows' of the `full' right momentum map of the fluid.

In Section~\ref{sec:collisions} we discuss the explicit dynamical behavior of the
particle model, in particular we study `collisions' leading to mergers of jetlets.
This makes an explicit connection between dynamics in the different
levels of the jetlet hierarchy. It can potentially be useful when
simulating jetlet systems: when two particles become close, they
can be replaced by a merged state whose momenta do not blow up.

We redo the analysis of Mumford and Michor to find that two 0-jetlets
can either merge in infinite time, or bounce off each other.
Then, we analyze the asymptotic dynamics of this merged state and show
that it coincides with the dynamics of a single 1-jetlet particle. This
improves the claim in \cite{CotterHolmJacobsMeier2014} which showed the
convergence to a 1-jetlet state without explicitly considering the dynamics.

Finally, we suggest a more algebraic interpretation of mergers in
the jetlet hierarchy by viewing all levels of the hierarchy as embedded
in the larger space $\mathfrak{X}_{\rm div}(\mathbb{R}^n)^*$, where levels form (part of) the
boundaries of other levels (a bit like in a CW-complex).

In summary, the numerics section~\ref{sec:Numerical experiments} shows first of all that it is feasible to implement this jetlet model numerically, also in dimensions higher than
two. Further, we corroborate the analytical results and confirm the
conserved quantities and jetlet particle merging behavior numerically.
Several different experiments show that the merging/scattering behavior persists
under various perturbations that cannot be studied analytically anymore.
It also shows how $J_L$ (i.e.\ angular and linear) momentum is exchanged
in the jet-particle collisions.

We provide a detailed appendix. In Appendix~\ref{sec:Hamiltonian}, a brief discussion of symplectic and Poisson manifolds is followed by the definition of weak dual pairs. Then we prove Theorem~\ref{thm:dual_pairs}, which is the core result used in the paper. Appendix \ref{app:diagramatic} presents an overview over the spaces used in the paper and relates them to some general results of reduction theory, see in particular Figures \ref{fig_diagram_1} and \ref{fig_diagram_2}. 
 In Appendix~\ref{sec:multi} we describe the index conventions used in the main text (e.g., when calculating momentum maps for the general jetlet solutions), while Appendix~\ref{sec:measure_valued_momap} briefly discusses the more abstract point of how the dual space of vector fields can be viewed as tensor product of 1-forms and distributions.
Finally, Appendix~\ref{sec:eom} provides information describing the equations of motion for 1-jets, but in reduced coordinates, which arise after the reduction that eliminates the conserved quantities, $J_R$.

\subsection{Previous work}
  Lagrangian models of ideal fluids such as \emph{smooth particle hydrodynamics} \cite{Monaghan1977,Lucy1977} and
\emph{vortex methods} \cite{Chorin1973}, do not exhibit structures which express the nested character of the diffeomorphism group.
 One means of obtaining a Lagrangian model with a nested structure
 was recently presented in \cite{JacobsRatiuDesbrun2013}
 where a sequence of infinite-dimensional reductions by symmetry was
 executed, to produce a hierarchy of finite-dimensional systems from a
 regularized fluid model.
 The finite-dimensional systems were particle-like solutions
 in which each particle carries a model of the diffeomorphism group,
 known as jet groups.
 Thus, \cite{JacobsRatiuDesbrun2013} derived a Lagrangian analog of 
 ``whirls within whirls'', similar to the Eulerian models proposed
 in \cite{HolmTronci2012}.
 However, the specifics of the regularized fluid model were not determined,
 and the analysis of \cite{JacobsRatiuDesbrun2013} was purely
 formal.

 Later, a regularized version of the ideal fluid equations was presented by
 \cite{MumfordMichor2013}.
 This new partial differential equation was amenable to the theory presented in 
 \cite{JacobsRatiuDesbrun2013}, and gave rise to a specific and easily
 implementable manifestation of the hierarchy of particle models
 described there.
 In \cite{CotterHolmJacobsMeier2014} we numerically computed some
 of these particle-like solutions and observed cascade phenomena as
 an emergent behavior at the zeroth level in the hierarchy.

 It is notable that the zeroth level of the hierarchy has been studied in the context of partial differential equations with hydrodynamic background.
 In particular, \cite{HoldenRaynaud2006} provide the first convergence proof of such a method in the context of the Camassa-Holm equation.
 This same idea was implemented for the $n$-dimensional Camassa-Holm
 equation in \cite{ChertockDuToitMarsden2012}.
 In the context of image registration algorithms, the need to obtain
 compressible diffeomorphisms motivated the use of particle methods
 in a similar framework \cite{JoshiMiller2000}.
 These methods, designed for a wide array of PDEs, are studied analytically
 in \cite{TrouveYounes2005}, which also contains a proof of well-posedness
 for a range of PDEs.
 Finally, in the context of image registration, \cite{Sommer2013} 
 discovered a compressible fluid version of the hierarchy derived in \cite{JacobsRatiuDesbrun2013},
 and numerically integrated solutions in the first level of the hierarchy.

\subsection{Notation}
We will let $\mathfrak{X}(\R^n)$ denote the space of $H^\infty$ vector fields,
and we let $\mathfrak{X}_{\rm div}(\R^n)$
denote the space of divergence free vector fields resulting from the Hodge decomposition.
We let $\Diff(\R^n)$ denote the space of $H^\infty$ diffeomorphisms of $\R^n$  (see \cite{MichorMumford2013})
and we let $\SDiff(\R^n)$ denote the subgroup of volume preserving diffeomorphisms.

Various different types of indices will be used throughout this paper.
To distinguish between the types, we keep the following
conventions:
\begin{itemize}
\item indices $a,b,c,\ldots$ label particles and range from $1$ to $N$;
\item indices $i,j,k,\ldots$ label space coordinates and range from
  $1$ to $n$, the dimension of space;
\item superscript indices $(k)$ denote the order of jets in a
  coordinate-free representation.
\end{itemize}

\section{Main approach}
\label{sec:approach}
Let us begin by describing our strategy for obtaining
particle-like solutions to regularized fluid equations.
Our main framework is that of Hamiltonian mechanics
and symplectic geometry.
In particular, our main hammer is Theorem~\ref{thm:dual_pairs} (see page
\pageref{thm:dual_pairs} in the Appendix), which we repeat
here for convenience.
  \begin{quote}
    {\bf Theorem \ref{thm:dual_pairs}}
    Let $P_1$ and $P_2$ be Poisson manifolds and let $S$ be
    a symplectic manifold.  Let $J_1,J_2:S \to P_1,P_2$ form a weak dual pair,
    see Definition~\ref{def:weak_dual_pair}.
    Let $h \in C^1(P_1)$.
    If $(q,p)(t) \in S$ is a solution to Hamilton's equations
    with respect to the Hamiltonian $H = h \circ J_1$,
    then $J_1\left( (q,p)(t) \right) \in P_1$ is a solution
    to Hamilton's equations on $P_1$ with respect to $h$,
    and $J_2( (q,p)(t))$ is constant in time.
  \end{quote}

We will leverage this theorem in the following way.
We will find a (weak) dual pair, $J_1,J_2: S \to \mathfrak{g}^*$, where
\begin{enumerate}
\item $S$ is the space of particle locations and momenta (and, for the higher orders of the hierarchy, internal group variables),
\item $\mathfrak{g}$ is the space of vector fields,
\item and $h$ is the Hamiltonian of a regularized model of ideal fluid (i.e.\ a kinetic energy).
\end{enumerate}
Theorem~\ref{thm:dual_pairs} applied to the (weak) dual pair $J_1,J_2$
then yields particle-like solutions and conserved quantities to the equations of motion
of a regularized model of an ideal incompressible fluid.

\section{Example: the rigid body}
\label{sec:rigid_body}
  In this section we review basic notions from classical mechanics
by studying the motion of a rigid body whose center of mass rests
at the origin.
We will see a first application of Theorem~\ref{thm:dual_pairs}
in this context.
The configuration of a rigid body is described by a
rotation matrix $R \in \SO(3)$.
The equations of motion are given by Hamilton's equations
on the cotangent bundle $T^*\SO(3)$.
There are canonical coordinates on $T^*\SO(3)$ given by $(R,P)$
where $P$ is such that $R^TP $ is a $3 \times 3$ anti-symmetric
matrix\footnote{%
  We use the pairing $\langle P, v \rangle = \tr(P^T v)$ for
  $P \in T^*_R SO(3)$ and $v \in T_R SO(3)$. More generally, the pairing $\langle \cdot , \cdot \rangle$ will always refer to the natural pairing between a vector space and its dual in this paper.}.
The angular momentum in the body frame is the
unique vector $\Pi \in \mathbb{R}^3 \cong \so(3)^*$ such that
\begin{align*}
  J_R(R,P) := R^TP = \begin{bmatrix}
    0 & -\Pi_3 & \Pi_2 \\
    \Pi_3 & 0 & -\Pi_1 \\
    -\Pi_2 & \Pi_1 & 0 
    \end{bmatrix}.
\end{align*}
Denoting this map from $(R,P)$ to $\Pi$ by $J_R$,
the Hamiltonian can be written as a function
of $\Pi$.  In particular, the reduced Hamiltonian is
\begin{align*}
  h(\Pi) = \frac{1}{2}\Pi \cdot \mathbb{I}^{-1} \cdot \Pi,
\end{align*}
where $\mathbb{I}$ is a non-degenerate $3\times 3$ symmetric matrix
known as the moment of inertia matrix.
The unreduced Hamiltonian is $H(R,P) = h(J_R(R,P))$.

  The group $\SO(3)$ acts upon itself by left multiplication.
  This action can be lifted to an action on $T^*\SO(3)$ given by
  \begin{align*}
    (R,P) \in T^* \SO(3) \stackrel{g \in \SO(3) }{\mapsto}
    (g R , gP ) \in T^* \SO(3).
  \end{align*}
  It is notable that $\Pi = R^T P = R^T g^T g P = (gR)^T (gP)$
  is unaltered by this transformation.
  Therefore the Hamiltonian, $h(\Pi)$, is invariant under this action of $\SO(3)$.
  By Noether's theorem there is a conserved quantity associated to 
  this symmetry.
  The conserved quantity is manifested by the \emph{momentum map}
  \begin{align*}
    J_L(R,P) = PR^T .
  \end{align*}
  Moreover, due to this symmetry, one can write the evolution equations
  on a lower-dimensional space.
  The very fact that the Hamiltonian is written in terms of $\Pi$
  suggests that the equations of motion can be written in terms of $\Pi$ alone.
  Indeed this is the case,
  \begin{align}
    \dot{\Pi} = \Pi \times (\mathbb{I}^{-1} \Pi ). \label{eq:rigid_body}
  \end{align}
  This equation can be seen as a Hamiltonian equation on $\mathbb{R}^3$
  with respect to the non-canonical Poisson bracket
  \begin{align*}
    \{ F , G \}_{\rm Nambu}(x)  = - x \cdot (\nabla F \times \nabla G) 
  \end{align*}
  known as the \emph{Nambu bracket}.
  This is no coincidence. Let us first introduce the isomorphism of
  Lie algebras $(\R^3,\times)$ and $\so(3)$, given by the so-called hat map
  \begin{equation}\label{eq:hat-map}
    x \in \R^3 \mapsto
    \hat{x} = 
    \begin{bmatrix}
      0   & -x_3 &  x_2 \\
      x_3 & 0    & -x_1 \\
     -x_2 &  x_1 & 0
    \end{bmatrix} \in \so(3).
  \end{equation}
  \begin{prop}[see {\cite[\S 2.5]{HolmBook2}}] \label{prop:Nambu}
    The Nambu bracket on $\mathbb{R}^3$
    is identified with the Lie--Poisson bracket on $\so(3)^*$
    through the hat map isomorphism~\eqref{eq:hat-map}
    in the sense that $\{ f,g\}_{\rm Nambu}( x) = \{ \hat{f} , \hat{g} \}_{\rm LP}( \hat{x})$
    where $f,g \in C^1(\mathbb{R}^3)$ and 
    $\hat{f},\hat{g} \in C^1(\so(3)^*)$ are defined by 
    $\hat{f}( \hat{x} ) = f(x), \hat{g}( \hat{x} ) = g(x)$.
  \end{prop}
  \begin{proof}
    The Lie--Poisson bracket on $\so(3)^*$ is
    \begin{align*}
      \{ \hat{f} , \hat{g} \}_{\rm LP} (\hat{\Pi}) =
      \left \langle \hat{\Pi} , \left[ d\hat{f}(\hat{\Pi}) , d\hat{g}(\hat{\Pi}) \right]
        \right \rangle,
    \end{align*}
    for arbitrary functions $\hat{f},\hat{g} \in C^1(\so(3)^*)$.
    There exist functions $f,g \in C^1(\mathbb{R}^3)$ related to $\hat{F},\hat{G}$ through the hat map.
    One can observe that $d\hat{f}(\hat{\Pi}) \in \so(3)$ is 
    related to $\nabla f(\Pi)$ through the relation $\widehat{ \nabla f(\Pi)} = d\hat{f}( \hat{\Pi})$.
    We see that the commutator bracket satisfies
    \begin{align*}
      [\hat{x},\hat{y} ] := \hat{x} \hat{y} - \hat{y} \hat{x} = \widehat{x \times y }.
    \end{align*}
    Therefore the Lie--Poisson bracket can be written as
    \begin{align*}
    \{ \hat{f} , \hat{g} \}_{\rm LP}(\hat{\Pi})
    &= \langle \hat{\Pi} , \widehat{ \nabla f \times \nabla g }(\Pi) \rangle \\
    &= \Pi \cdot  \left( \nabla f \times \nabla g \right)(\Pi) \\
    &= \{ f , g \}_{\rm Nambu}(\Pi).
    \end{align*}
  \end{proof}

  Recall that $T^*\SO(3)$ is a symplectic manifold. The momentum maps
  $J_L,J_R$ arise canonically from the left and right action of
  $\SO(3)$ on $T^* \SO(3)$. The actions commute and it can also
  be checked that $J_L$ and $J_R$ have symplectically orthogonal
  kernels, hence it follows from~\cite[Corollary~2.6]{GayBalmazVizman2012}
  that the diagram
  \begin{alignat*}{2}
    \so(3)^* \stackrel{J_L}{\longleftarrow}&
    \,T^* \SO(3)
    &&\stackrel{J_R}{\longrightarrow} \so(3)^* \\
     PR^T \stackrel{J_L}{\longmapsfrom}&
    \;\; (R,P)
    &&\stackrel{J_R}{\longmapsto} R^TP
  \end{alignat*}
  is a dual pair
  (see page \pageref{thm:dual_pairs} for details).
  The maps $J_L$ and $J_R$ are called \emph{symplectic variables} in \cite{MarsdenWeinstein1983}
  as they allow one to pull-back calculations on a Poisson manifold to a symplectic manifold.\footnote{
  	\cite{MarsdenWeinstein1983} also referred to $J_R$ and $J_L$ as \emph{Clebsch variables},
		however this terminology has changed over the past few decades.}

  By Theorem~\ref{thm:dual_pairs}, this dual pair expresses rigid
  body dynamics and conserved quantities.
  Specifically, the right leg yields the reduced phase space where 
  the system evolves in time.
  The left leg yields the conserved quantities of the rigid body
  associated with the left action of $\SO(3)$ on itself.
  The most important aspect of these maps is that they are both Poisson
  maps, i.e.\ they carry the canonical Poisson bracket on $T^* \SO(3)$
  to the Nambu bracket on $\so(3)^* \cong (\mathbb{R}^3,\times)$,
  as the following proposition shows.

  \begin{prop}[remark 2.5.11 \cite{HolmBook2}] \label{prop:SO3_to_Nambu}
    Let $\{ \cdot , \cdot \}_{\rm can}$ denote the canonical Poisson bracket
    on the cotangent bundle $T^* \SO(3)$.
    Let $\{ \cdot , \cdot \}_{\rm Nambu}$ denote the Nambu bracket on $\mathbb{R}^3$.
    Both $J_L$ and $J_R$ are Poisson maps.
    Explicitly, this means
    \begin{align*}
      -\{ f \circ J_L , g \circ J_L \}_{\rm can} = \{ f , g \}_{\rm Nambu} \circ J_L \\
      \{ f \circ J_R , g \circ J_R \}_{\rm can} = \{ f , g \}_{\rm Nambu} \circ J_R \\
    \end{align*}
    for any $f,g \in C^2( \mathbb{R}^3)$.
  \end{prop}
  \begin{proof}
    Let $(R,P) \in T^* \SO(3)$ and set $\Pi = J_R(R,P)$.
    We observe that
    \begin{align*}
      \langle \Pi , \Omega \rangle
      = \tr( \hat{\Pi}^T \hat{\Omega} ) 
      = \tr( P^T R \hat{\Omega} ) 
      = \langle P , R \cdot \hat{\Omega} \rangle.
    \end{align*}
    This tells us that $J_R$ is the momentum map associated
    with the cotangent lift of the right action of $\SO(3)$ on itself.
    Such momentum maps are always equivariant
    and thus yield Poisson maps (Theorem 12.4.1 \cite{MandS}).
    Thus $J_R$ carries the canonical
    Poisson bracket on $T^* \SO(3)$, to the Lie--Poisson
    bracket on $\so(3)^*$.
    By Proposition~\ref{prop:Nambu}, this is nothing but the
    Nambu bracket upon identifying $\so(3)$ with $\mathbb{R}^3$.
    The same argument applies to $J_L$ using a left action.
  \end{proof}

  One can obtain solutions to \eqref{eq:rigid_body}
  by solving canonical Hamiltonian equations with respect to $H(R,P)$.
  In particular, if $(R,P)(t)$ is a solution to Hamilton's equation,
  then $\Pi(t) = J_R( (R,P)(t))$ is a solution to Hamilton's equation
  with respect to the Nambu bracket.
  This is a result of Proposition~\ref{prop:Poisson_dynamics}
  paired with the observation that $J_R$ is a Poisson map via
  Proposition~\ref{prop:SO3_to_Nambu}.

\section{Regularized fluids}
\label{sec:reg_fluids}
  Euler's equations of motion for incompressible fluids can be seen as
  Hamiltonian equations on the (dual) space of divergence
  free vector fields \cite{Arnold1966}.
  Consider the Lie algebra of vector fields on $\mathbb{R}^n$,
  denoted by $\mathfrak{X}(\mathbb{R}^n)$.
  Formally, the dual space to $\mathfrak{X}(\mathbb{R}^n)$ is a Poisson
  manifold when equipped with the Lie-Poisson bracket
  (see \eqref{eq:Lie-Poisson} in Appendix~\ref{sec:Poisson}).
  We may consider the map
  $\psi: T^*R^n \to \mathfrak{X}(\mathbb{R}^n)^*$
  given implicitly by
  \begin{align*}
    \langle \psi(q,p) , u \rangle = p \cdot u(q) \quad
    \forall u \in \mathfrak{X}(\mathbb{R}^n), (q,p) \in T^*\mathbb{R}^n.
  \end{align*}
  Explicitly we may write $\psi$ using the Dirac-delta
  functional as $\psi(q,p) = p \otimes \delta_q$.
  It is shown in \cite{HolmMarsden2005} that this map is Poisson.
  Furthermore, as the $n$-dimensional Camassa--Holm equation \cite{CamassaHolm1993} is a
  Hamiltonian equation on $\mathfrak{X}(\mathbb{R}^n)^*$,
  Proposition~\ref{prop:Poisson_dynamics} promises to express a 
  certain subset
  of solutions by solving Hamiltonian equations for a finite number
  of particles.
  Specifically, $\psi$ yields the peakon solutions of the $n$-dimensional
  Camassa--Holm equation.
  In this section we explore analogous constructions for
  an incompressible and regularized version of the
  Camassa--Holm equation, discovered in \cite{MumfordMichor2013}.

%

In the case where $h_{\rm Euler}(m) = \frac{1}{2} \| m \|^2_{L^2}$ is the standard fluid kinetic energy on the
dual vector space to the incompressible vector fields, $\mathfrak{X}_{\rm div}(\mathbb{R}^n)^*$, 
Hamilton's equations are written as
\begin{align}
	\partial_t m+ \lie_u [ m] = 0 \quad , \quad u^i = \delta^{ij} m_j. \label{eq:ideal_fluid}
\end{align}
where $\lie_u[m]$ is the Lie derivative of $m$.
The primary finding of \cite{Arnold1966} was that \eqref{eq:ideal_fluid} is equivalent to the
inviscid fluid equation
\begin{align*}
	\partial_t u + u \cdot \nabla u = - \nabla p \quad , \quad  \nabla \cdot u = 0.
\end{align*}
Since we have not yet clarified the Poisson structures of the system, it may not be obvious that \eqref{eq:ideal_fluid} is a Hamiltonian equation. The following proposition shows this for a general Hamiltonian  on $\mathfrak{X}_{\rm div}(\R^n)^*$.

  \begin{prop}[\cite{Arnold1966}] \label{prop:LPDiff}
  Let $h \in C^{\infty}( \mathfrak{X}_{\rm div}(\R^n)^* )$.
  Recall that $\mathfrak{X}_{\rm div}(\R^n)^*$ is a Poisson manifold
  when equipped with the Lie--Poisson bracket,
  and given a function $h$, the Fr\'echet derivative $dh(m)$ is an element  of $\mathfrak{X}_{\rm div}(\R^n)^{**}$.
  In the event that $dh(m) \in \mathfrak{X}_{\rm div}(\R^n)$,
  Hamilton's equations are given by
  \begin{align*}
    \dot{m} + \lie_u [m] = 0 \quad , \quad u = dh(m). 
  \end{align*}
\end{prop}
\begin{proof}
  To each $v \in \mathfrak{X}_{\rm div}(\R^n)$
  we can associate a linear function on $\mathfrak{X}_{\rm div}(\R^n)^*$
  given by $m \in \mathfrak{X}_{\rm div}(\R^n)^* \mapsto \langle m , v \rangle \in \R$.
  Let us denote this function by $f_v$.
  Let $m(t)$ satisfy Hamilton's equations.
  By the definition of the Lie--Poisson bracket
  (see \eqref{eq:Lie-Poisson} in Appendix~\ref{sec:Poisson})
  we observe
  \begin{align*}
    \frac{d}{dt} f_v(m) &= \{ f_v , h \}(m) = \langle m , [ df_v(m) , dh(m) ] \rangle \\
    &= \langle m , \lie_{dh(m)}[ df_v(m) ] \rangle.
  \end{align*}
  However $df_v = v$ and so
  the last line can be equated with  $-\langle \lie_{dh(m)}[m] , v \rangle$.
  Additionally, we know that $\frac{d}{dt} f_v(m) = \langle \dot{m} , v \rangle$ since $v$ is constant in time.
  Therefore we find
  \begin{align*}
    \langle \dot{m} + \lie_{dh(m)} [m] , v \rangle = 0.  
  \end{align*}
  As $v$ is arbitrary, this uniquely characterizes $\dot{m}$.\footnote{This is \emph{not} a ``weak'' characterization.  The entity $\dot{m}$ is
  contained in the dual space to $\mathfrak{X}_{\rm div}(\R^n)$ and it is
  therefore \emph{defined uniquely} by how it acts on $\mathfrak{X}_{\rm div}(\R^n)$.}
  The result follows.
\end{proof}

\subsection{The Mumford--Michor model}
\label{sec:MME}
Consider the Hamiltonian
\begin{align*}
  h_{p,\sigma}(m) = \frac{1}{2} \langle m , K_{p,\sigma} * m \rangle_{L^2},
\end{align*}
where $K_{p,\sigma}:\R^n \to \mathbb{R}^{n \times n}$ is the matrix
valued Green's kernel defined by the property
\begin{align*}
 \Big(1 - \frac{\sigma^2}{p} \Delta \Big)^p \cdot \int_{\R^n} K_{p,\sigma}^{ij}(x - y) m_j(y) dy = \delta^{ij} m_j(x).
\end{align*}
In this case Hamilton's equations take the form
\begin{align}
	\partial_t m+ \lie_u [ m] = 0 \quad , \quad u^i  = K^{ij}_{p,\sigma} * m_j. \label{eq:MMDiff}
\end{align}
Solutions to \eqref{eq:MMDiff} exhibit existence and uniqueness for all time.
Moreover, as $\sigma \to 0$, $h_{p,\sigma} \to h_{\rm Euler}$ and one
can speculate that solutions to~\eqref{eq:MMDiff}
approach solutions to the ideal fluid equation~\eqref{eq:ideal_fluid}.
In fact, this is the case over short times, and for $\sigma > 0$ solutions of \eqref{eq:MMDiff}
differ from those of \eqref{eq:ideal_fluid} by an amount $\sigma t$
in the $H^k$-norm.
Thus Hamilton's equations with respect to $h_{p,\sigma}$ have been proposed
as a model for ideal fluids \cite[Theorems 2 and 3]{MumfordMichor2013}. From now on we shall often suppress the
parameters $p,\sigma$ and shorten $K = K_{p,\sigma}$ to prevent index
clutter.

Next, we discuss a dual pair for this system.
The group $\SDiff(\R^n)$ acts on itself from the left and from the right.
These actions can be lifted to $T^* \SDiff(\R^n)$, and yield
momentum maps $J_L,J_R\colon T^*\SDiff(\R^n) \to \mathfrak{X}_{\rm div}(\R^n)^*$. In particular these maps form the dual pair
\begin{align*}
  \mathfrak{X}_{\rm div}(\R^n)^*
  \stackrel{J_L}{\longmapsfrom}
  T^* \SDiff(\R^n)
  \stackrel{J_R}{\longmapsto}
  \mathfrak{X}_{\rm div}(\R^n)^*.
\end{align*}
For a mathematically rigorous treatment of this dual pair we refer to~\cite{GayBalmazVizman2012}.

By Theorem~\ref{thm:dual_pairs}, we can use this dual pair to derive dynamical properties of Hamiltonian equations defined on $\mathfrak{X}_{\rm div}(\R^n)^*$.
The Hamiltonian for a fluid is written on the left instance of $\mathfrak{X}_{\rm div}(\R^n)^*$.
One can (in principle) solve Hamilton's equations on $T^*\SDiff(\R^n)$
with respect to the Hamiltonian $H = h \circ J_L$.
This yields the \emph{material} or \emph{Lagrangian} coordinate perspective
of fluid mechanics one encounters in a first course on continuum
mechanics.
The right leg yields conserved quantities associated with the particle relabeling symmetry of the fluid.
It was found in \cite{Arnold1966} that these conserved momenta are identical
to the law of conservation of circulation, that is, \emph{Kelvin's circulation theorem}.

Unfortunately, this dual pair does not help us in solving~\eqref{eq:MMDiff}
since solving Hamilton's equations on $T^*\SDiff(\R^n)$ is no less difficult
than solving Hamilton's equations on $\mathfrak{X}_{\rm div}(\R^n)^*$.
In the next section we will derive a dual pair wherein the symplectic
manifold is more reasonable.
This will yield the particle-like solutions described in \cite{MumfordMichor2013}. In the later parts of the paper we will generalize the treatment to obtain a \emph{weak} dual pair for each level in the hierarchy of particle-like solutions.

\subsection{Particle-like solutions}
\label{sec:zeroth_order}
There is a natural left group and algebra action of $\SDiff(\R^n)$
and $\mathfrak{X}_{\rm div}(\R^n)$, respectively, on $\R^n$ given by
\begin{align*}
  q \in \R^n
  &\stackrel{ \varphi \in \SDiff(\R^n) }{\longmapsto}
  \varphi(q) \in \R^n \\
  q \in \R^n
  &\stackrel{ u \in \mathfrak{X}_{\rm div}(\R^n) }{\longmapsto}
  u(q) \in T_q\R^n.
\end{align*}
The tangent lift of the former is defined in the obvious way by sending
\begin{align*}
  (q,v) \in T\R^n \cong \R^n \times \R^n
  \stackrel{ \varphi \in \SDiff(\R^n) }{\longmapsto}
  \big(\varphi(q) , (\partial_j\varphi^i(q) v^j) \partial_i \big) \in T_{\varphi(q)} \R^n.
\end{align*}
The cotangent lift is defined by taking the dual
of  the tangent lifted action.
That is to say,
\begin{align*}
  (q, p) \in T^*\R^n \cong \R^n \times \R^n
  \stackrel{ \varphi \in \SDiff(\R^n) }{\longmapsto}
  \big(\varphi^{-1}(q) , (\partial_i\varphi^j(q) p_j) e^i \big),
\end{align*}
where $\{e^i\}$ forms the dual basis to $\{\partial_i\}$ at $\varphi^{-1}(q)$.
The momentum map
$J^{(0)}_L : T^*\R^n \to \mathfrak{X}_{\rm div}(\R^n)^*$
associated to this left action is defined by the
condition
\begin{align*}
  \langle J_L^{(0)}( q , p) , u \rangle := \langle p , u(q) \rangle,
\end{align*}
for all $u \in \mathfrak{X}_{\rm div}(\R^n)$ and $(q,p) \in T^*\R^n$ (see \eqref{eq:cotangent_momap} in Appendix~\ref{sec:Symplectic}); the superscript in $J_L^{(0)}$ serves as a reminder that we are considering the zeroth level in the hierarchy of particle-like solutions.
We see that $J_L^{(0)}(q,p)$ is an evaluation operator, and we can write it 
more explicitly as a measure-valued momentum map in terms of the Dirac-delta distribution as
\begin{align*}
  J_L^{(0)}( q , p ) = p \otimes \delta_q,
\end{align*}
where $\delta_q$ is the Dirac-delta distribution on $\R^n$ centered at $q$.
This identification holds modulo $dC^1(\R^n) \otimes d{\bf x}$ where $dC^1(\R^n)$
is the space of $C^0$-exact one-forms and $d{\bf x}$ is the canonical volume form on $\R^n$
(see Appendix~\ref{sec:measure_valued_momap}).
Since $J_L^{(0)}$ is a cotangent lift momentum map, it is equivariant and therefore Poisson, again by \cite[Theorem~12.4.9]{MandS}.

We define the manifold for particles,
\begin{align*}
  Q_N^{(0)} = \left\{ (q_1,\dots,q_N) \in \R^n \times \dots\times \R^n
                 \mid q_a \neq q_b \text{ when } a \neq b \right\}.
\end{align*}
We will index the particles with $a,b,c,\dots$ and Cartesian coordinate
directions in $\R^n$ with indices $i,j,k,\dots$
Thus each $q \in Q_N^{(0)}$ can be decomposed
into $N$ particles as $(q_1,\dots,q_N)$, where each $q_a \in \R^n$
and the $i$-th coordinate of the $a$-th particle
is denoted by $q\indices{_a^i}$.

The group $\SDiff(\R^n)$ acts on $Q_N^{(0)}$ by the diagonal action.
Through the same manipulations as we applied previously we obtain the
momentum map for $N$ particles given by
\begin{align*}
  J_L^{(0)}(q,p) = p_a \otimes \delta_{q_a},
\end{align*}
where a sum over repeated indices is implied.

By Propositions~\ref{prop:Poisson_dynamics} and~\ref{prop:LPDiff},
we obtain solutions to Hamilton's equations
on $\mathfrak{X}_{\rm div}(\R^n)^*$, by solving
Hamilton's equations on $T^*\R^n$ if $dh( J_L^{(0)}(q,p) )$ is
a vector field.  If $h = h_{p,\sigma}$, then we calculate that $dh$
evaluated on the $J_L^{(0)}(p,q)$ is the vector field
\begin{align*}
	dh( p_a \otimes \delta_{q_a} )  (x) = K^{ij}(x - q_a) p_{a\,j} \pder{}{x^i}.
\end{align*}
This is a vector field whose differentiability is determined completely by that of
the kernel, $K$.
\begin{figure}
	\centering
	\includegraphics[width = 0.4\textwidth]{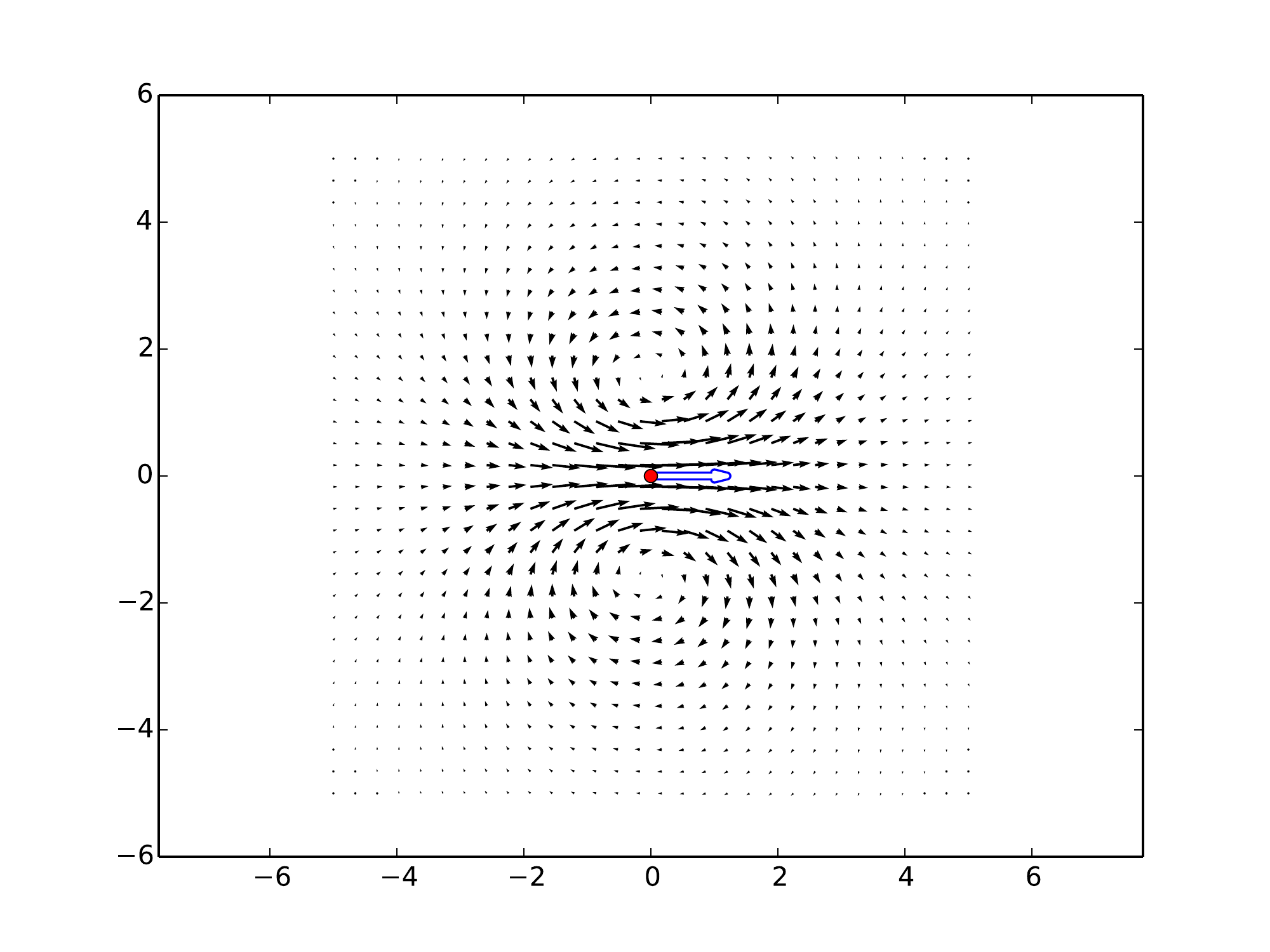}
	\includegraphics[width = 0.4\textwidth]{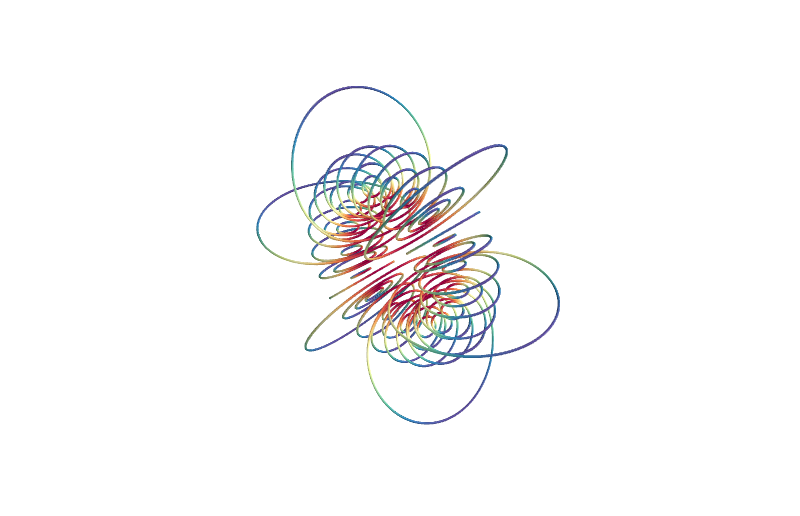}
	\caption{A $0$-jetlet with momentum $m = (1,0)^T \otimes \delta_0$ in dimensions $2$ and $3$.
	(\emph{left}) Quiver plot of the induced 2 dimensional velocity field.
	(\emph{right})  Streamline plot of the induced 3 dimensional velocity field. }
	\label{fig:zero_jetlet}
\end{figure}
Once one has found a solution $(q(t), p(t))$ to Hamilton's equations on $T^*\mathbb{R}^n$ and thus also a solution to Hamilton's equations on $\mathfrak{X}_{\rm div}(\mathbb{R}^n)^*$, one can proceed to integrate the corresponding time-dependent vector field to obtain the fluid motion $\varphi_t$ in $\SDiff(\mathbb{R}^n)$. It is natural (but not mandatory) to choose the initial map $\varphi_0$ to be the identity or at least to be a diffeomorphism that satisfies $\varphi_0(q_a(0))= q_a(0)$ for all $a$.
  If this choice is made, one can interpret the curve $q(t) = (q_1(t), \ldots, q_N(t))$ in $Q_N^{(0)}$ as the locations of particles as they are swept along by the fluid flow, that is, $q_a(t) = \varphi_t(q_a(0))$ for all $a = 1, \ldots, N$. 
  
  We see that $J_L^{(0)}$ is injective, and thus has a trivial kernel. As a result, the symplectic orthogonal to the kernel of $J_L^{(0)}$  is the full tangent bundle $T(T^*Q_N^{(0)})$. Hence, if we define the (trivial) map $J_R^{(0)}: T^*Q^{(0)}_N \to \mathfrak{X}_{\rm div}(\mathbb{R}^n)^*$ by $J_R^{(0)} \equiv 0$, it follows that the diagram
 \begin{alignat*}{2}
    \mathfrak{X}_{\rm div}(\R^n)^* \stackrel{J_L^{(0)} }{\longleftarrow}&
    \,T^* Q_N^{(0)}
    &&\stackrel{J_R^{(0)} }{\longrightarrow} \mathfrak{X}_{\rm div}(\R^n)^* \\
    p_a \otimes \delta_{q_a} \stackrel{J_L^{(0)}}{\longmapsfrom}&
    \;\;(q,p)
    &&\stackrel{J_R^{(0)} }{\longmapsto} 0
  \end{alignat*}
  is a (proper) dual pair.
  This dual pair allows us to express conservation
  laws and dynamics as a result of Theorem~\ref{thm:dual_pairs}.
  Namely, the left leg represents the space in which particle-like
  solutions to \eqref{eq:MMDiff} evolve,
  while the right leg represents a (trivial) conserved quantity.

  In order to make contact with the later parts of the paper, it is useful to remark that $J_R^{(0)}$ can formally be understood as the cotangent lift momentum map associated with a certain (trivial) group action.
  To that end, we fix a designated point $z = (z_1, \ldots, z_N) \in Q_N^{(0)}$ and take an arbitrary element $ q= (q_1, \ldots, q_N) \in Q_N^{(0)}$ to represent the set of all $\varphi \in \SDiff(\mathbb{R}^n)$ that satisfy $\varphi \cdot z  = q$.
  That is, $\varphi(z_a) = q_a$ for all $a$.
  This means in particular that the specification of $(q_1, \ldots, q_N)$ fixes the zeroth order Taylor expansion (at the locations $z_a$) of the corresponding set of diffeomorphisms. With this in mind, let us define the isotropy group
\begin{align}
\iso(z)  = \{ \psi \in \SDiff(\mathbb{R}^n) | \psi(z_a) = z_a \mbox{ for all } a\} \label{eq:isotropy_group}
\end{align}
and a (trivial) right action on $Q_N^{(0)}$ where $\psi \in \iso(z)$ maps an element $(\varphi(z_1), \ldots, \varphi(z_N))$ to $( \varphi \circ \psi (z_1), \ldots, \varphi \circ \psi(z_N))$.
Similar constructions will be crucial in later sections, when
constructing the right leg of an --- in this case weak --- dual pair for the higher levels in the hierarchy of particle-like solutions.

  In summary we find:
  \begin{prop}[\S7 of \cite{MumfordMichor2013}] \label{prop:0-solutions}
  Let $p \geq \frac{n}{2} +1$ and $\sigma > 0$.
  Let $H:T^*Q_N^{(0)} \to \R$ be the function
  \begin{align*}
    H(q,p) =\frac{1}{2} K_{p,\sigma}^{ij}(q_a - q_b)\, p_{a\,i}\, p_{b\,j}.
  \end{align*}
  If $(q,p)(t)$ is a solution to Hamilton's equations, then
  $m(t) = p_a(t) \otimes \delta_{q_a(t)}$
  is a solution to \eqref{eq:MMDiff}.
\end{prop}
\begin{proof}
	Note that $H = h_{p,\sigma} \circ J_L^{(0)}$ and apply Theorem~\ref{thm:dual_pairs}.
  See the remark after Theorem~\ref{thm:k-solutions} for more details on the kernel smoothness condition on $p$.
\end{proof}

\subsection{First order particle-like solutions}
\label{sec:first_order}
In this section we revisit the first order particle-like solutions of \cite{CotterHolmJacobsMeier2014} and discuss their weak dual pair, before extending the treatment to the higher levels of the hierarchy in the subsequent section. 
  Let $\SL(n)$ denote the Lie group of $n\times n$ matrices
  with unit determinant.
  Let $q = (q^{(0)}, q^{(1)}) \in \R^n \times \SL(n)$ and consider the left
  $\SDiff(\R^n)$ action on $\R^n \times \SL(n)$ given by
  \begin{align}
    \varphi \cdot q = (\varphi(q^{(0)} ) , D\varphi(q^{(0)} ) \cdot q^{(1)} ). \label{eq:first_order_action}
  \end{align}
  Where $D\varphi(q^{(0)} ) \cdot q^{(1)}$ is the result of multiplying
  the Jacobian matrix $D\varphi(q^{(0)} )$ with $q^{(1)}$.
  \begin{prop}
    The action of $\SDiff(M)$ on $\R \times \SL(n)$ in \eqref{eq:first_order_action} is a group action.
  \end{prop}
  \begin{proof}
    Since $\varphi \in \SDiff(\R^n)$, it follows that $D\varphi |_{q^{(0)}} \in \SL(n)$.
    Therefore $D\varphi |_{q^{(0)}} \cdot q^{(1)} \in \SL(n)$.
    Secondly, if $\varphi_1,\varphi_2 \in \SDiff(\R^n)$ we observe that
    \begin{align*}
      \varphi_2 \cdot (\varphi_1 \cdot q)
      &= \varphi_2 \cdot \big(\varphi_1(q^{(0)} ) , D\varphi_1 |_{q^{(0)}} \cdot q^{(1)} \big) \\
      &= \big(\varphi_2(\varphi_1(q^{(0)} )) , \left. D\varphi_2 \right|_{\varphi_1( q^{(0)} )} \cdot D\varphi_1 |_{ q^{(0)} } \cdot q^{(1)} \big) \\
      &= \big( (\varphi_2 \circ \varphi_1)(q^{(0)} ) , D( \varphi_2 \circ \varphi_1)|_{q^{(0)} } \cdot q^{(1)} \big) \\
      &= (\varphi_2 \circ \varphi_1) \cdot q,
    \end{align*}
    where the second and third lines are applications of the chain rule.
  \end{proof}
  As before, this action can be lifted to the cotangent bundle $T^*(\R^n \times \SL(n))$.  Specifically, 
  the action of the diffeomorphism $\varphi^{-1}$ is given by
  \begin{equation}\label{eq:SDiff_cotangent_action_Q1}
    (T\varphi)^* \cdot ( q^{(0)} , q^{(1)}  , p^{(0)} , p^{(1)} ) \\
    = \big( \varphi^{-1}(q^{(0)})  , [D\varphi|_{q^{(0)}}]^{-1} \cdot q^{(1)} , D\varphi|_{q^{(0)}}^* \cdot p^{(0)} , D\varphi|_{q^{(0)}}^* \cdot p^{(1)} \big).
  \end{equation}
  Also as before, we can generalize this construction to the space of
  $N > 1$ particles by considering the space
  \begin{align*}
    Q^{(1)}_N = \left\{  ( q_1 , \dots, q_N ) \,\Bigg|
      \begin{array}{c}
        q_a = (q^{(0)}_a,q^{(1)}_a) \in \R^n \times \SL(n) \\
        q^{(0)}_a \neq q^{(0)}_b \text{ when } a \neq b
      \end{array} \right\}.
  \end{align*}
  For convenience it is nice to choose coordinates at this point.
  If we let $q_a$ denote the position of the $a$-th particle,
  $q\indices{_a^i}$ denote the $i$-th component of this position,
  and let $q\indices{_a^i_j}$ denote the $(i,j)$ entry of the $a$-th
  matrix then the resulting momentum map is defined by the condition
  \begin{align*}
    \langle J_L^{(1)}(q,p) , u \rangle
    = p\indices{_a_i} u^i(q^{(0)}_a) + p\indices{_a_i^j} \partial_k u^i(q^{(0)}_a) q\indices{_a^k_j}
  \end{align*}
  for an arbitrary $u \in \mathfrak{X}_{\rm div}(\R^n)$.
  In terms of the Dirac-delta functional, we can write $J_L^{(1)}$ as the measure-valued momentum map
  \begin{align*}
    J_L^{(1)}(q, p)
    = p\indices{_a_i} \dx^i \otimes \delta_{q_a^{(0)}}
     -p\indices{_a_i^j} q\indices{_a^k_j} \dx^i \otimes \partial_k \delta_{q_a^{(0)}}.
  \end{align*}
  Next, we construct a dual momentum map associated with a right action on $Q_N^{(1)}$. Analogous to the previous section  it is useful at this stage to write elements of $Q_N^{(1)}$ in the form $ \varphi \cdot z = (\varphi (z_1, {\mathbf{1}}), \ldots, \varphi \cdot (z_N , \mathbf{1}))$, $\varphi \in \SDiff(\mathbb{R}^n)$, for some designated element $ z=( (z_1, \mathbf{1}), \ldots, (z_N, \mathbf{1}))$. Clearly, every element of $Q_N^{(1)}$ can be written in this form for some  $\varphi \in \SDiff(\mathbb{R}^n)$. Indeed, with this convention the specification of an element in $Q_N^{(1)}$ fixes the first order Taylor expansion of $\varphi$ at the locations $z_a$, $a = 1, \ldots, N$. With this in mind, let us recall the isotropy group $\iso(z)$ defined earlier in \eqref{eq:isotropy_group}, which leaves these locations invariant, and define a right action on  $Q_N^{(1)}$ given by 
\begin{equation*}
	 (\varphi \cdot z) \cdot \psi  = (\varphi \circ \psi) \cdot z.
\end{equation*}
That is,
\begin{equation}
	 (q^{(0)}, q^{(1)}) \cdot \psi = \left( \left(q_1^{(0)}, q_1^{(1)} \cdot D\psi|_{z_1}\right), \ldots, \left(q_N^{(0)}, q_N^{(1)} \cdot D\psi|_{z_N}\right)\right). \label{1-rightaction}
\end{equation}
In Proposition~\ref{HL_action_prop} we will generalize this construction to define a right action for the higher levels in the hierarchy of particle-like solutions.
  The action \eqref{1-rightaction} yields the cotangent lift momentum map defined by the condition
  \begin{align*}
    \langle J_R^{(1)}(q,p) , u \rangle = p\indices{_a_i^j} q\indices{_a^i_k} \partial_j u^k(z_a).
  \end{align*}
  In terms of the Dirac-delta functional we may write this as the measure-valued momentum map
  \begin{equation}\label{eq:JR1}
    J_R^{(1)}(q,p) = - p\indices{_a_i^j} q\indices{_a^i_k} \dz^k \otimes \partial_j \delta_{z_a}
  \end{equation}
  
  \begin{prop}\label{prop:J1-dual-pair}
    The momentum maps $J_L^{(1)}$ and $J_R^{(1)}$ form a weak dual pair.
  \end{prop}
   We postpone the proof, as this is a special case of a proposition which comes later in the paper (Proposition \ref{prop:k dual pairs}).

  As before, the quantity $dh( J_L^{(1)}(q,p))$ is a legitimate vector field
  if $K_{p,\sigma}$ is sufficiently smooth.  Moreover, when the Hamiltonian
  $H^{(1)} = h_{p,\sigma} \circ J_L^{(1)}$ is $C^1$ we may evolve Hamilton's equations
  to obtain solutions.

  \begin{prop} \label{prop:1-solutions}
    Let $p \ge \frac{n}{2} + 2$ and $\sigma > 0$.
    Then $H^{(1)} = h_{p,\sigma} \circ J_L^{(1)}$ is $C^1$ and given by the expression
	\begin{align*}
		H^{(1)}(q,p) &= \frac{1}{2} p\indices{_a_i} K^{ij}(q_a - q_b) p\indices{_b_j}\\
			&\quad + p\indices{_a_i^l} q\indices{_a^k_l} (\partial_k K^{ij})(q_a - q_b) p\indices{_b_j} \\
			&\quad -  \frac{1}{2} p\indices{_a_i^n} q\indices{_a^l_n} (\partial_l \partial_k K^{ij})(q_a - q_b) q\indices{_b^k_m} p\indices{_b_j^m}.
  \end{align*}
  If $(q,p)(t) \in T^*Q_N^{(1)}$ is a solution to Hamilton's equations with respect to
	$H^{(1)}$, then $J_L^{(1)}( (q,p)(t)) \in \mathfrak{X}_{\rm div}(\R^n)^*$ is
	a solution of \eqref{eq:MMDiff}, and $J_R^{(1)}( (q,p)(t)) \in \mathfrak{X}_{\rm div}(\R^n)^*$ is
	constant in time. 
  \end{prop}
  The proof of the above proposition is identical to that of
  Proposition~\ref{prop:0-solutions}. See again the remark following
  Theorem~\ref{thm:k-solutions} for more details on the kernel
  smoothness condition on $p$.

  As before, it is useful to interpret the trajectory $q(t)$ of the previous proposition in terms of the curve $\varphi_t \in  \SDiff(\mathbb{R}^n)$ obtained by integrating the time-dependent vector field $dh_{p, \sigma} J_L^{(1)}((q, p)(t)) = K_{p, \sigma} * J_L^{(1)}((q, p)(t))$. If one chooses $\varphi_0$ to be the identity or  any other element of $\SDiff(\mathbb{R}^n)$ that satisfies $\varphi_0 \cdot (q(0)) = q(0)$, then  $q(t) = \varphi_t \cdot  q(0)$. This implies in particular that the $q_a^{(0)}(t)$ are the trajectories of the particles as they are swept along by the fluid flow.

  Various vector fields for large kernel smoothness $p$
  are depicted in Figure~\ref{fig:zoo}
  for different initial values of the traceless matrix
  $\mu\indices{_i^j} = p\indices{_i^l} q\indices{^j_l}$.

  \begin{figure}
  	\centering
        \includegraphics[width=0.3\textwidth]{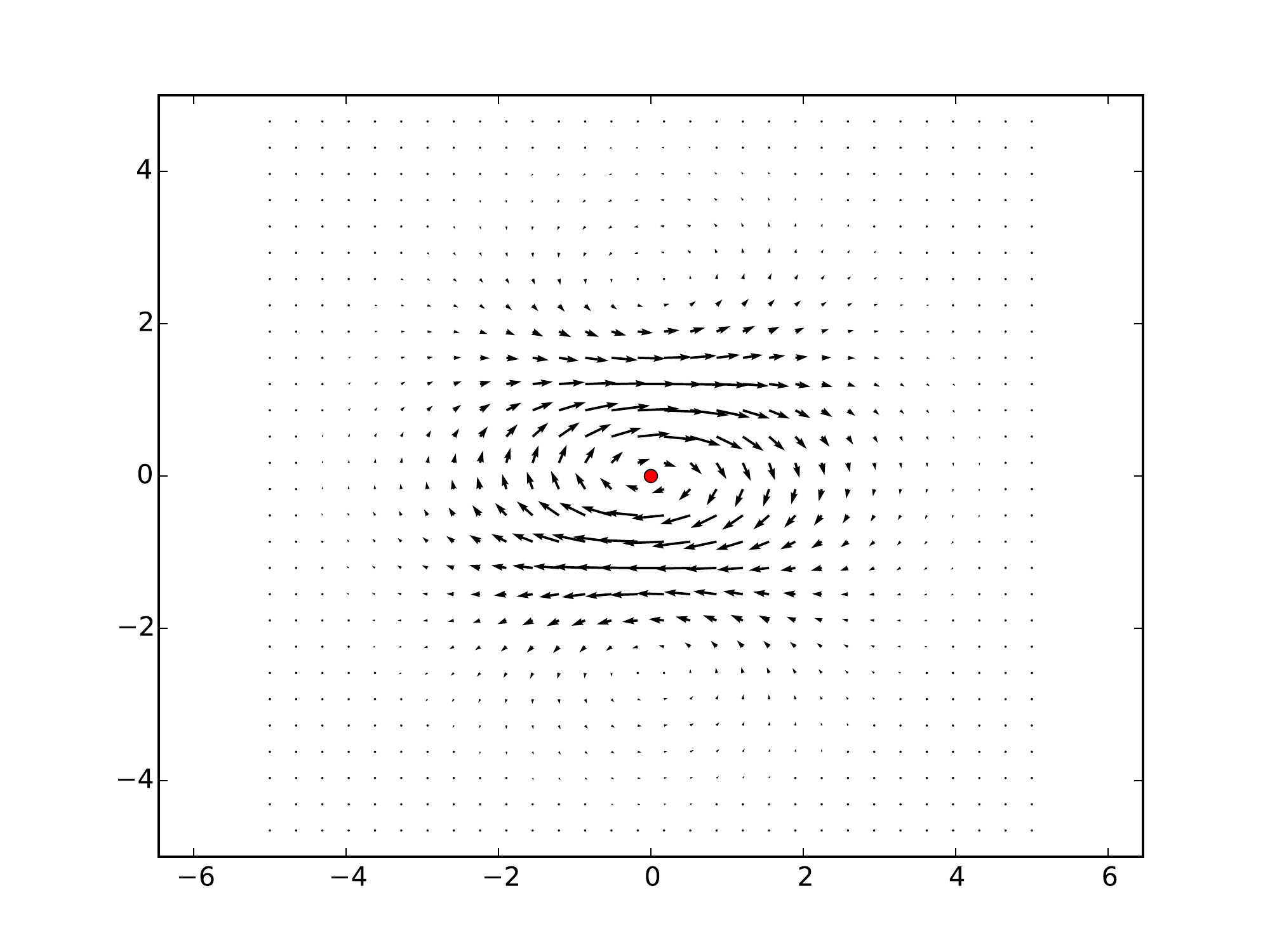}
        \includegraphics[width=0.3\textwidth]{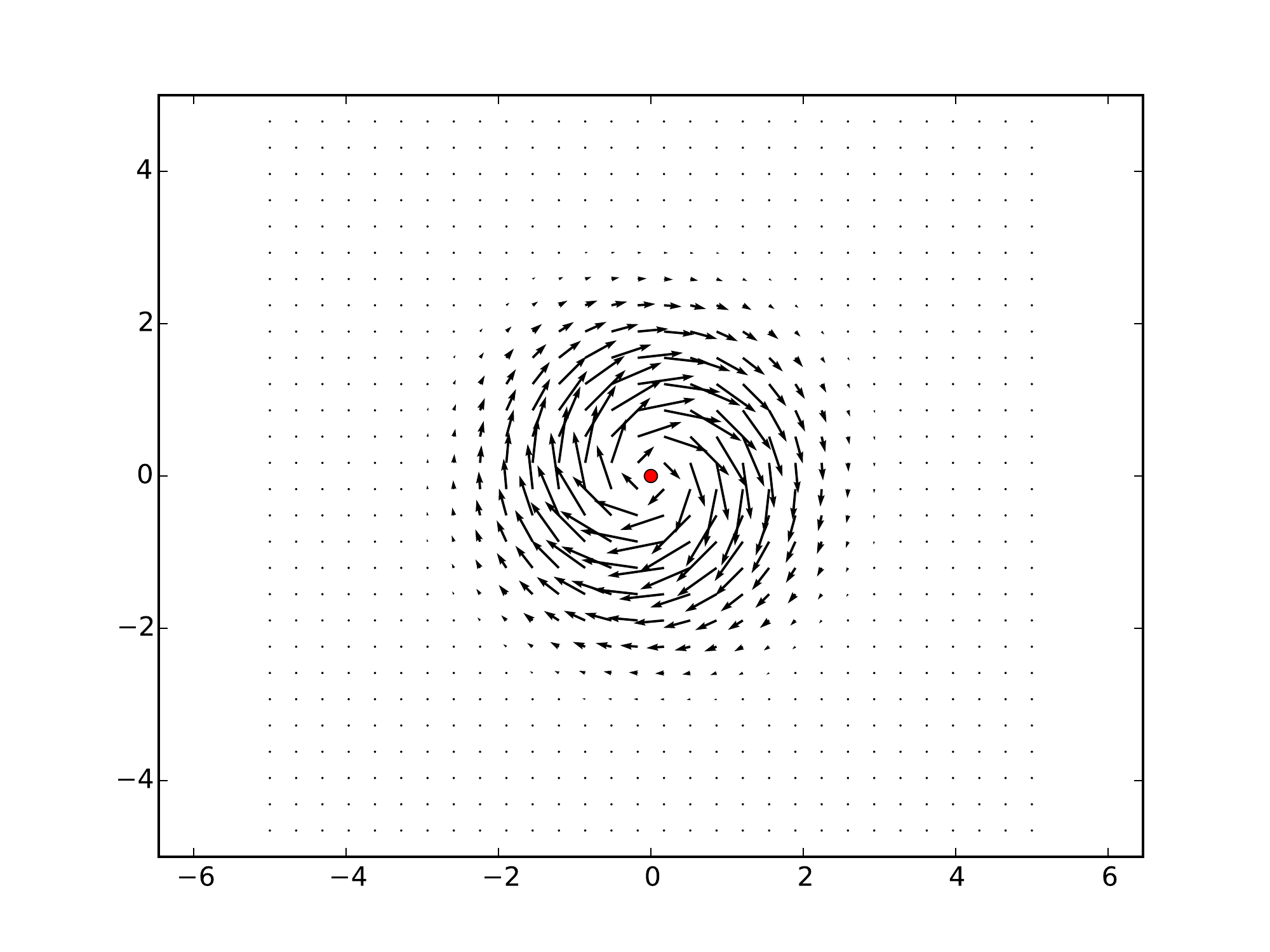}
        \includegraphics[width=0.3\textwidth]{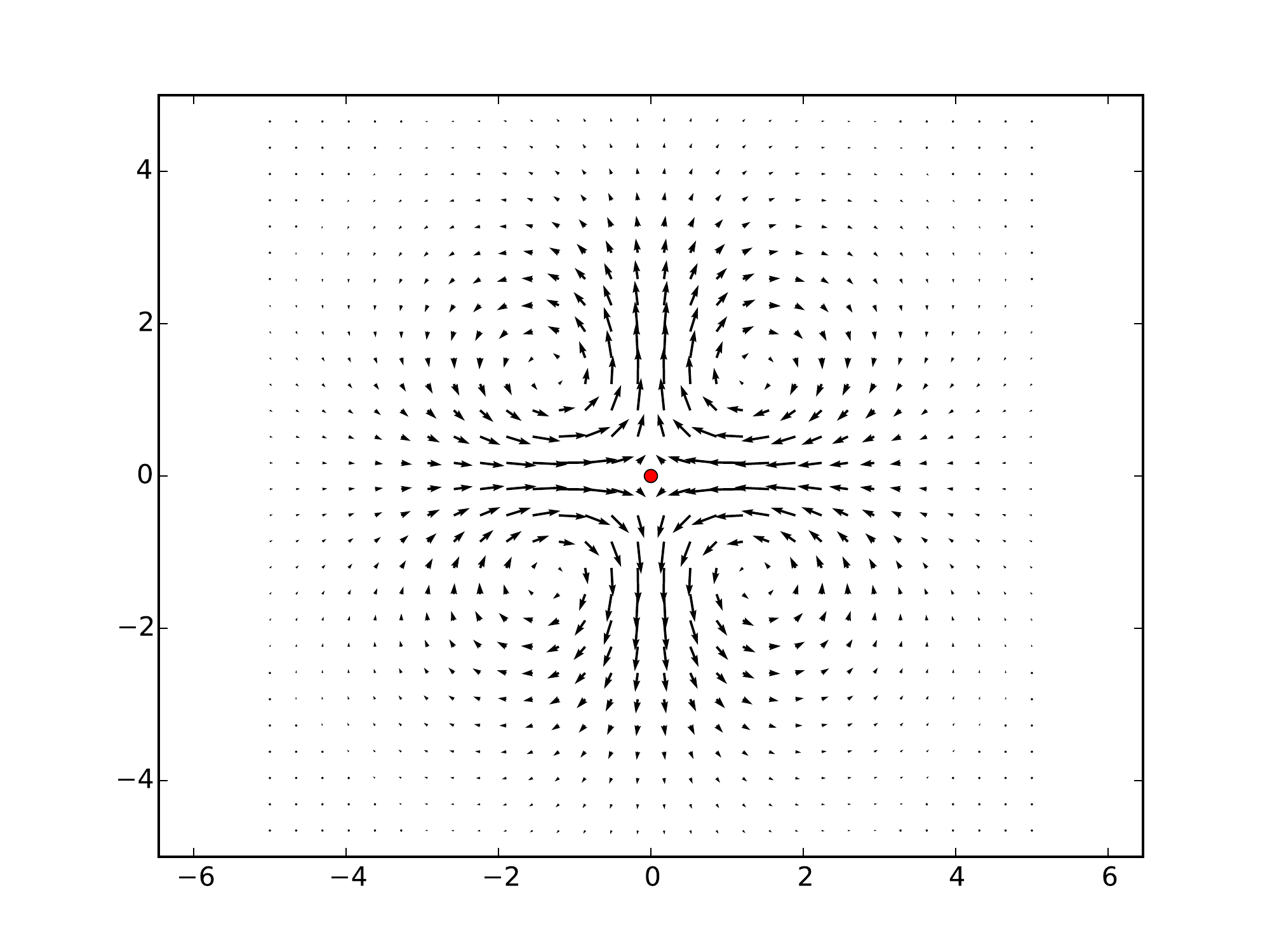}
        \\
        \includegraphics[width=0.4\textwidth]{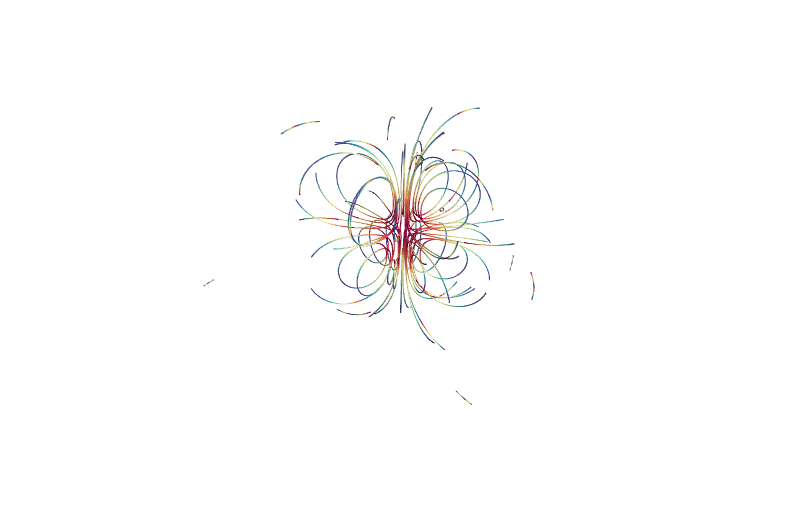}
        \includegraphics[width=0.4\textwidth]{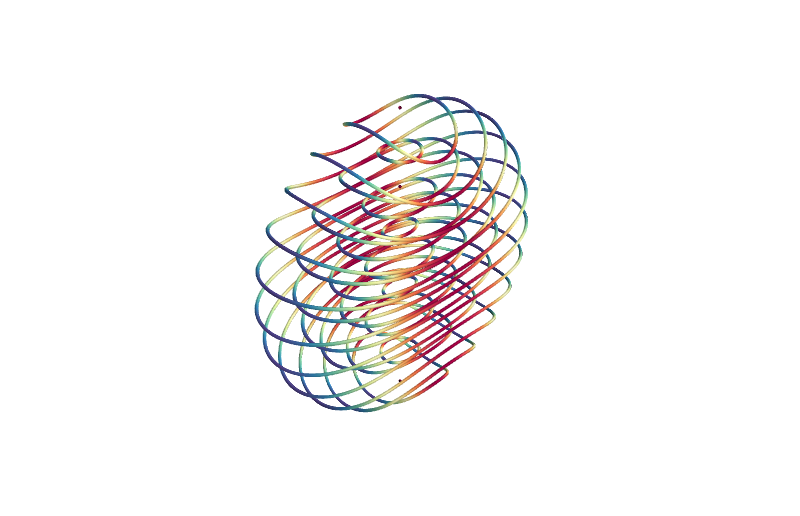}
        \caption{(\emph{top}) Quiver plots of 2D vector fields induced by first order jetlets. 
        (\emph{bottom})   Streamline plots of 3D vector fields induced by first order jetlets.}
        \label{fig:zoo}
  \end{figure}

  \subsection{Higher order particles}
  \label{sec:higher_order}
  In this section we will introduce a hierarchy of particle-like solutions whose $k$-th level includes the solutions at the $(k-1)$-th level.
  The zeroth level in the hierarchy consists of the standard particle-like solutions,
  while the first level describes the particles with internal $\SL(n)$ variables
  discussed in the previous section.
  The particles in the $k$-th level carry the coefficients of
  $k$-th order Taylor expansions of diffeomorphisms,
  or \emph{jets}.
  Therefore, we call these particles \emph{$k$-jetlets}.

  Let $\varphi \in \SDiff(\R^n)$.
  The zeroth order Taylor expansion of $\varphi$ about $0$
  is $\varphi(0)$, and the collection of such Taylor
  expansions is all of $\R^n$.
  The first order Taylor expansion of $\varphi$ about $0$ is
  \begin{align*}
    \varphi^i( x) = \varphi^i(0) + \partial_j \varphi^i(0) x^j + o( \|x\|).
  \end{align*}
  The tuple of coefficients $(\varphi(0) , D\varphi(0) ) \in \R^n \times \SL(n)$ is
  called the first order jet of $\varphi$ evaluated at $0$.
  Going further, the second order Taylor expansion of $\varphi$ about $0$
  is
  \begin{align*}
    \varphi^i(x) = \varphi^i(0) + \partial_j \varphi^i(0) x^j + 
    \frac{1}{2} \partial_{jk} \varphi^i(0) x^j x^k + o( \| x\|^2).
  \end{align*}
  We see that $D^2\varphi(0)$ is a tensor of rank $(1,2)$, which is
  symmetric in the lower indices.
  We call the space of such tensors $S^1_2$,
  and we see that the space of second order jets is a submanifold $Q^{(2)}_1 \subset \R^n \times \SL(n) \times S^1_2$.
  Finally, for $k \geq 2$ the space of $k$-th order jets is a submanifold
  \begin{align*}
    Q_1^{(k)} \subset \R^n \times \SL(n) \times S^1_2 \times S^1_3 \times \cdots\times S^1_k,
  \end{align*}
  where $S^1_j$ is the vector space of $(1,j)$-tensors which
  have been symmetrized in the covariant indices.
  The space $Q_1^{(k)}$ is equipped with the fiber bundle
  projection $\pi^{(k)} : Q_1^{(k)} \to \R^n$,
  which projects onto the $\R^n$ component.
  In fact, $Q_1^{(k)}$ is a trivial principle bundle
  where fibers, contained within $\SL(n) \times S^1_2 \times \cdots \times S^1_k$
  form a jet group \cite{JacobsRatiuDesbrun2013}.
  The jet group serves as a finite-dimensional model of the diffeomorphism group (see \cite[Chapter 4]{KMS99} for a description of the group multiplication),
  and this motivates our interpretation of jetlets as models of self-similarity.
  We define the space for an $N$-tuple of $k$-jetlets by taking a product
  \begin{align*}
    Q^{(k)}_N = \{ (q_1,\dots, q_N) \in Q_1^{(k)} \times \dots \times Q_1^{(k)}
    \mid \pi^{(k)}(q_a) \neq \pi^{(k)}(q_b) \text{ when } a \neq b \}.
  \end{align*}
  We coordinatize $Q^{(k)}_N$ as follows.
  We will use Greek indices to represent spatial
  multi-indices on $\R^n$ (see Appendix~\ref{sec:multi} for our multi-index convention).
  A typical coordinate on $Q^{(k)}_N$ will therefore look
  like $q\indices{_a^i_\beta}$ where $a \in \{1,\dots,N\}$, $i \in \{ 1 , \dots, n \}$
  and $\beta$ is a multi-index on $\R^n$.
  This coordinate is used to model the partial derivative of the $i$-th coordinate
  of a diffeomorphism at some point $z_a \in \R^n$, i.e.\ $\partial_\beta \varphi^i(z_a)$.
  For the statement of the next proposition, recall that the definition  of $\iso(z)$ was given earlier in \eqref{eq:isotropy_group}. Note also that we write $\Jet^k_x: \SDiff(\R^n) \to Q_1^{(k)}$ for
    the function which evaluates the spatial derivatives
    of a diffeomorphism up to order $k$ at the location $x \in \R^n$. In the subsequent sections we will also use the obvious generalization to multiple locations $z_1, \ldots, z_N \in \R^n$, which we will denote by $\Jet^k_z : \SDiff(\R^n) \to Q_N^{(k)}$.  

  \begin{prop}
    The group $\SDiff(\R^n)$ acts on $Q_N^{(k)}$ by a left Lie group
    action.  If $z_1,\dots,z_N \in \R^n$ are distinct points,
    then the isotropy group $\iso(z) \subset \SDiff(\R^n)$
    acts on $Q_N^{(k)}$ by a right Lie group action. \label{HL_action_prop}
  \end{prop}
  \begin{proof}
    Let $\varphi_1,\varphi_2 \in \SDiff(\R^n)$.
    The partial derivatives of $(\varphi_1 \circ \varphi_2)$
    are given by the Fa\`a di Bruno formula
    \begin{align*}
      \partial_\alpha( \varphi_1 \circ \varphi_2) = \sum_{k=1}^{|\alpha|}
      \left(
        \sum_{
        \substack{
          j_1,\dots,j_k \in \{1,\dots,n\} \\
          [\gamma_1,\dots,\gamma_k] \in \Pi(\alpha , k)
          }
          }
          \partial_{j_1 \cdots j_k} \varphi_1(\varphi_2(x))
          \partial_{\gamma_1} \varphi_2^{j_1} \cdots
          \partial_{\gamma_k} \varphi_2^{j_k}
        \right),
    \end{align*}
where we wrote $\Pi(\alpha, k)$ for the set of $k$-th order partitions of a multi-index $\alpha$. We refer to Appendix~\ref{sec:multi} for the details of our index conventions and to \cite{ConstantineSavits1996,Jacobs2014b} for a
    precise description of the multivariate Fa\`a di Bruno formula.
    One can read off   from the expression that
    a $k$-th order derivative only depends on $k$-th 
    and lower order
    partial derivatives of $\varphi_1$ and $\varphi_2$.
        A left $\SDiff(\R^n)$ action is induced on $Q_1^{(k)}$ by
    setting $\varphi \cdot q = \Jet^k_z( \varphi \circ \psi)$
    for any $\psi$ such that $\Jet^k_z(\psi) = q$.
    That this is independent of the choice of $\psi$ follows
    from observing that the Fa\`a di Bruno formula
    only uses data in $q = \Jet^k_z(\psi)$ and nothing more.
    In local coordinates, this action takes the form
    \begin{align*}
    (\varphi \cdot q)\indices{^i_\alpha} =
     \sum_{k=1}^{|\alpha|}
      \left(
        \sum_{
        \substack{
          j_1,\dots,j_k \in \{1,\dots,n\} \\
          [\gamma_1,\dots,\gamma_k] \in \Pi(\alpha , k)
          }
          }
          \partial_{j_1 \cdots j_k} \varphi^i(q^{(0)})
          q\indices{^{j_1}_{\gamma_1}} \cdots
          q\indices{^{j_k}_{\gamma_k}}
        \right),
    \end{align*}
    except for the component where $|\alpha | = 0$ in which case we observe $(\varphi \cdot q)^i = \varphi^i(q^{(0)})$.
    By the same construction, a right $\iso(z)$ action is induced
    on $Q_1^{(k)}$ by setting $q \cdot \varphi = \Jet^k_z(\psi \circ \varphi)$ for any $\psi$ such that $\Jet^k_z(\psi) = q$.
    We can choose distinct points $z_1,\dots,z_N \in \R^n$
    and apply the same process to $Q_N^{(k)}$.  In this case we observe the action to be
    \begin{align*}
      (q \cdot \varphi)\indices{^{ai}_\alpha} = \sum_{k=1}^{|\alpha|}
      \left(
        \sum_{
        \substack{
          j_1,\dots,j_k \in \{1,\dots,n\} \\
          [\gamma_1,\dots,\gamma_k] \in \Pi(\alpha , k)
          }
          }
          q\indices{^{ai}_{[j_1,\dots,j_k]}}
          \partial_{\gamma_1} \varphi^{j_1}(z_a) \cdots
          \partial_{\gamma_k} \varphi^{j_k}(z_a)
        \right).
    \end{align*}
  \end{proof}
  
  Just as in the previous sections, the actions of $\SDiff(\R^n)$ and $\iso(z)$, which commute, lift to actions on $T^*Q_N^{(k)}$.
  The associated momentum maps are given implicitly by how they act on the respective Lie algebras.
  In particular, the left action of $\SDiff(\R^n)$ yields the momentum map
  \begin{equation*}
  	\langle J_L^{(k)}(q,p) , u \rangle =
	\sum_{a=1}^N \left( p_{am} u^m(q^{(0)}_a) +
    \sum_{\substack{ |\alpha| \leq k \\ 1 \leq \ell \leq |\alpha| }} \;
		\sum_{
			\substack{
				j_1,\dots,j_\ell \in \{1,\dots,n\} \\
				[\gamma_1,\dots,\gamma_\ell] \in \Pi(\alpha,\ell)
				}
			}
			p\indices{_a_m^\alpha} \partial_{j_1 \cdots j_\ell}u^m (q^{(0)}_a)
      q\indices{_a^{j_1}_{\gamma_1}} \cdots
      q\indices{_a^{j_\ell}_{\gamma_\ell}}
			\right)
  \end{equation*}
  for arbitrary divergence free vector fields $u \in \mathfrak{X}_{\rm div}(\R^n)$.
  Equivalently, we may define $J_L^{(k)}$ as the unique map such that
  \begin{equation}\label{eq:moma_left}
    \langle J_L^{(k)}(q,p) , u \rangle = \langle (q,p) , \Jet_z^{k}( u \circ \varphi) \rangle
  \end{equation}
  for any $\varphi \in \SDiff(\R^n)$ whose $k$-jet is given by $q$ for any $(q,p) \in T^*Q^{(k)}_N$.

  The right action of $\iso(z)$ yields the momentum map $J_R^{(k)}$ defined implicitly by the relation
  \begin{equation*}
    \langle J_R^{(k)}(q,p) , w \rangle = \sum_{a=1}^N
    \sum_{\substack{ |\alpha| \leq k \\ 1 \leq \ell \leq |\alpha| }} \;
		\sum_{
			\substack{
				j_1,\dots,j_\ell \in \{1,\dots,n\} \\
				[\gamma_1,\dots,\gamma_\ell] \in \Pi(\alpha,\ell)
				}
			}
			p\indices{_a_i^\alpha} q\indices{_a^{i}_{[j_1 \cdots j_\ell]}}
      \sum_{m=1}^\ell  \partial_{\gamma_m} w^{j_m}(z_a) \left( \prod_{n \neq m}  \delta_{\gamma_n}^{[j_n]} \right),
  \end{equation*}
  where $w \in \mathfrak{X}_{\rm div}(\R^n)$ is such that $w(z) = 0$ (this describes the Lie algebra of $\iso(z)$),
  and $\delta_\alpha^\beta$ is the natural generalization of the Kronecker-delta symbol to multi-indices.
  Equivalently, we may define $J_R^{(k)}$
  as the unique map such that
 \begin{align}
  \langle J_R^{(k)}(q,p) , w \rangle = \langle (q,p) , \Jet_z^{k}( T\varphi \cdot w) \rangle, \label{Moma_right}
   \end{align}
  for any $\varphi \in \SDiff(\R^n)$ whose $k$-jet is given by $q$ for any $(q,p) \in T^*Q^{(k)}_N$. Here, we wrote $T\varphi \cdot w$ for the function obtained by applying the differential of $\varphi$ to $w$.
  
  We can write $J_L^{(k)}$ and $J_R^{(k)}$ explicitly, using the Dirac-delta distribution, as
  \begin{multline*}\label{eq:JLk}
    J_L^{(k)}(q,p) =
    \sum_{a=1}^N \Bigg(
    p_{am} \dx^m \otimes \delta_{q^{(0)}_a} \\
  + \sum_{\substack{ |\alpha| \leq k \\ 1 \leq \ell \leq |\alpha| }}
    (-1)^\ell p\indices{_{am}^\alpha} \dx^m \otimes
      \sum_{\substack{
          j_1,\dots,j_\ell \in \{1,\dots,n\} \\
          [\gamma_1,\dots,\gamma_\ell] \in \Pi( \alpha,\ell)
          }
        }
      q\indices{_a^{j_1}_{\gamma_1}} \cdots q\indices{_a^{j_\ell}_{\gamma_\ell}}
      \partial_{j_1 \cdots j_\ell} \delta_{q^{(0)}_a}
      \Bigg)
  \end{multline*}
  and
  \begin{equation}\label{eq:JRk}
    J_R^{(k)}(q,p) = \sum_{a=1}^{N}
    \sum_{\substack{ |\alpha| \leq k \\ 1 \leq \ell \leq |\alpha| }} \;
		\sum_{
			\substack{
				j_1,\dots,j_\ell \in \{1,\dots,n\} \\
				[\gamma_1,\dots,\gamma_\ell] \in \Pi(\alpha,\ell)
				}
			}
			p\indices{_a_i^\alpha} q\indices{_a^i_{[j_1 \cdots j_\ell]}}
      \sum_{m=1}^\ell (-1)^{|\gamma_m|}\left( \prod_{n \neq m}  \delta_{\gamma_n}^{[j_n]} \right) \dz^{j_m} \otimes \partial_{\gamma_m}\delta_{z_a}.
  \end{equation}

  \begin{prop} \label{prop:k dual pairs}
  	The maps $J_L^{(k)}$ and $J_R^{(k)}$ form a weak dual pair.
  \end{prop}
  \begin{proof}
  	 The result is a direct application of \cite[Corollary 2.8]{GayBalmazVizman2012}. To make the exposition more self-contained, we provide some details. By \cite[Corollary 2.6]{GayBalmazVizman2012} we need only show that $J_L$ and $J_R$ are equivariant, and that $J_L$ is invariant under the right action of $\iso(z)$.
	Equivariance follows from the fact that $J_R$ and $J_L$ are derived from cotangent lifted group actions \cite[Corollary 4.2.11]{FOM}.
	So we need only illustrate that $J_L$ is (right) $\iso(z)$ invariant. This can be seen as a consequence of the commutativity of the left and right actions on $Q^{(k)}$.  
	For notational clarity, let us denote this right action by $\rho : Q^{(k)} \times \iso(z) \to Q^{(k)}$. 
	Explicitly, any element $q \in Q^{(k)}$ is expressible as the $k$-jet of some diffeomorphism $\varphi$,
	and $\rho( q , \psi) = \Jet_z^{k}( \varphi \circ \psi)$.
	
	Similarly, any element of the tangent fiber $T_\varphi \SDiff(\R^n)$ may be written as a composition $u \circ \varphi$ for some $u \in \mathfrak{X}_{\rm div}(\R^n)$.
	Elements of $TQ^{(k)}$ are of the form $\Jet_z^{k}( u \circ \varphi)$ for $u \in \mathfrak{X}_{\rm div}(\R^n)$ and $\varphi \in \SDiff(\R^n)$.
	Given this representation, the tangent lift of the action $\rho$, also denoted $\rho: TQ^{(k)} \times \iso(z) \to TQ^{(k)}$, is given by
	\begin{align}
		\rho( (q,v)  , \psi ) = \Jet_z^{k}( u \circ \varphi \circ \psi ) \label{eq:tangent action}
	\end{align}
	for $(q,v) \in TQ^{(k)}$ and where $u \in \mathfrak{X}_{\rm div}(\R^n)$ and $\varphi \in \SDiff( \R^n)$ are arbitrary up to the constraint $(q,v) = \Jet_z^k( u \circ \varphi)$.
	The cotangent lifted action is a left action, $\rho^* : \iso(z) \times T^*Q^{(k)} \to T^*Q^{(k)}$, defined implicitly by the condition
	\begin{align}
		\langle \rho^*( \psi , (q,p) ) , \rho( (q,v) , \psi^{-1} ) \rangle = \langle (q,p) , (q,v) \rangle \label{eq:dual action}
	\end{align}
	for all $(q,p) \in T^*Q^{(k)}$ and $(q,v) \in TQ^{(k)}$.
	This action is equivalent to the one defined in the discussion preceding \cite[Corollary 4.2.11]{FOM}.
	In particular, $\rho^*( \psi , (q,p) )$ is a covector over the point $\tilde{q} = \rho( q , \psi^{-1})$.
	By the definition of $J_L$ we observe
	\begin{align*}
		\langle J_L( \rho^*( \psi , (q,p) ) ) , u \rangle = \langle  \rho^*( \psi ,  (q,p) ) , \Jet^{k}_z ( u \circ \bar{\varphi} ) \rangle, 
	\end{align*}
	where $\bar{\varphi} \in \SDiff(\R^n)$ is any diffeomorphism such that $\Jet_z^{k}( \bar{\varphi} ) = \rho( q , \psi^{-1} )$.
	If we let $\varphi$ be such that $\Jet^k_z( \varphi ) = q$ then we can simply choose $\bar{\varphi} = \varphi \circ \psi^{-1}$.
	Thus we get
	\begin{align*}
			\langle J_L( \rho^*( \psi , (q,p) ) ) , u \rangle &= \langle  \rho^*( \psi ,  (q,p) ) , \Jet^k_z ( u \circ \varphi \circ \psi^{-1} ) \rangle \\
			&= \langle  \rho^*( \psi ,  (q,p) )  ,  \rho( \Jet_z^{k}(u \circ \varphi) , \psi^{-1} ) \rangle \\
			&= \langle (q,p) , \Jet_z^{k}( u \circ \varphi) \rangle \\
			&= \langle J_L(q,p) , u \rangle, 
	\end{align*}
	using~\eqref{eq:tangent action} and~\eqref{eq:dual action} in the second and third equality.
	Thus, we see that $J_L$ is invariant under the right action of $\iso(z)$ on $Q^{(k)}$ and the result follows.
  \end{proof}
  
  In the case of $k=1$ we obtain the weak dual pair of the previous section and Proposition \ref{prop:J1-dual-pair} is a corollary of Proposition \ref{prop:k dual pairs}.
  Proposition \ref{prop:k dual pairs} gives us the final result on
  jetlet parametrized solutions.
  \begin{thm}\label{thm:k-solutions}
    Let $p \ge \frac{n}{2} + k + 1$ and $\sigma > 0$.
    Then $H^{(k)} = h_{p,\sigma} \circ J_L^{(k)}$ is $C^1$.
    Let $x(t)$ be a solution to Hamilton's equations on
    $T^*Q^{(k)}_N$, then $J_L^{(k)}( x(t))$ is a solution to Hamilton's
    equations on $\mathfrak{X}_{\rm div}(\R^n)^*$
    and $J_R^{(k)}( x(t))$ is constant in time.
  \end{thm}
  Remark that $K_{p,\sigma}$ has smoothness $C^{2p-n-1}$ by considering
  that the Fourier representations of its derivatives are integrable,
  while we need $K_{p,\sigma} \in C^{2k+1}$ to allow composition with
  $k$-th derivatives of delta distributions on each side and still
  obtain a $C^1$ Hamiltonian.

  As before, an even richer family of velocity fields is generated by a single particle of this type.
  At order $k=2$ this yields four new varieties of velocity fields per particle.
  Two examples of such velocity fields are depicted in Figure~\ref{fig:2_jet}.
  
  \begin{figure}[h!]
  	\centering
	\includegraphics[width=0.4\textwidth]{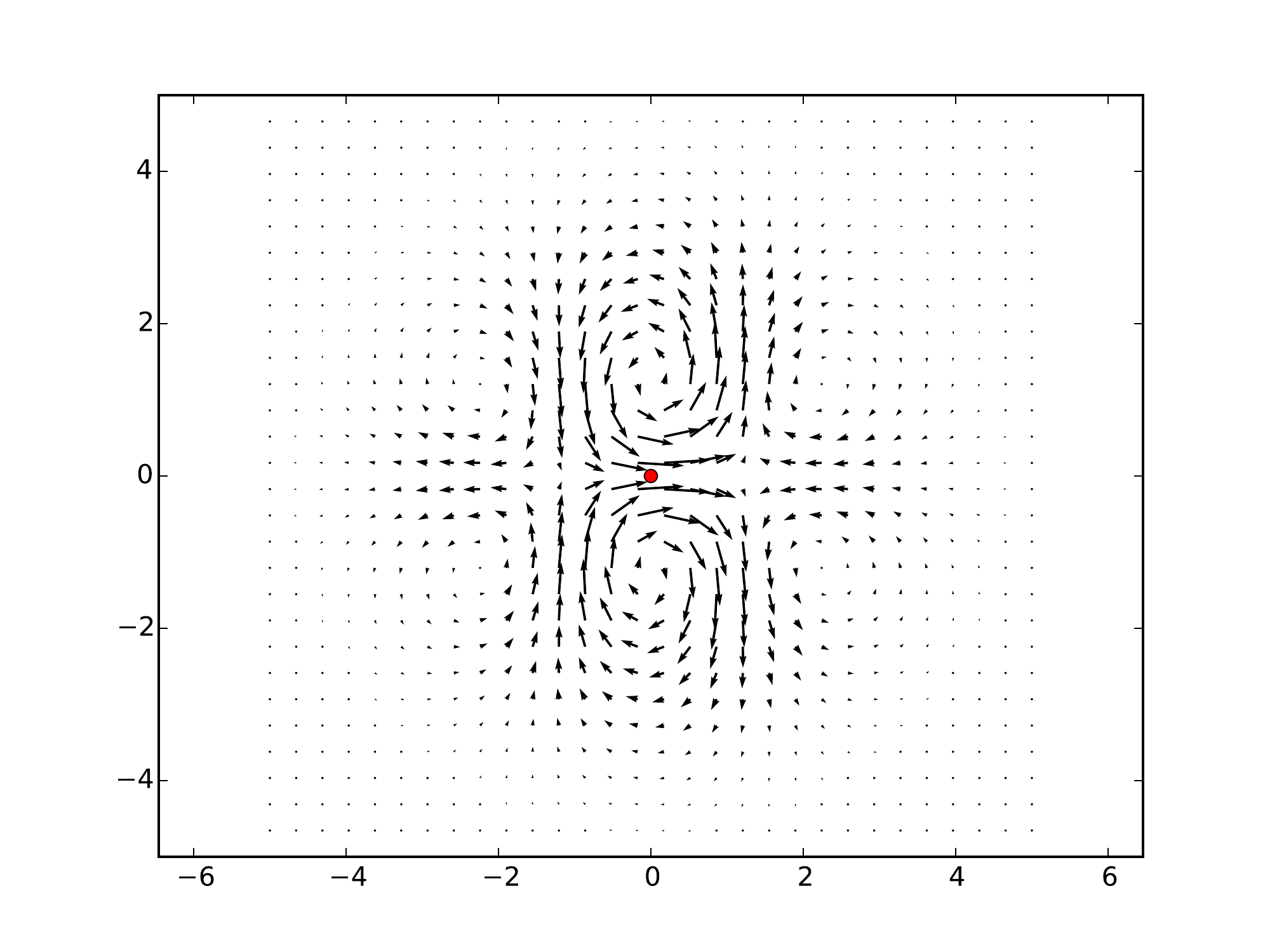}
	\includegraphics[width=0.4\textwidth]{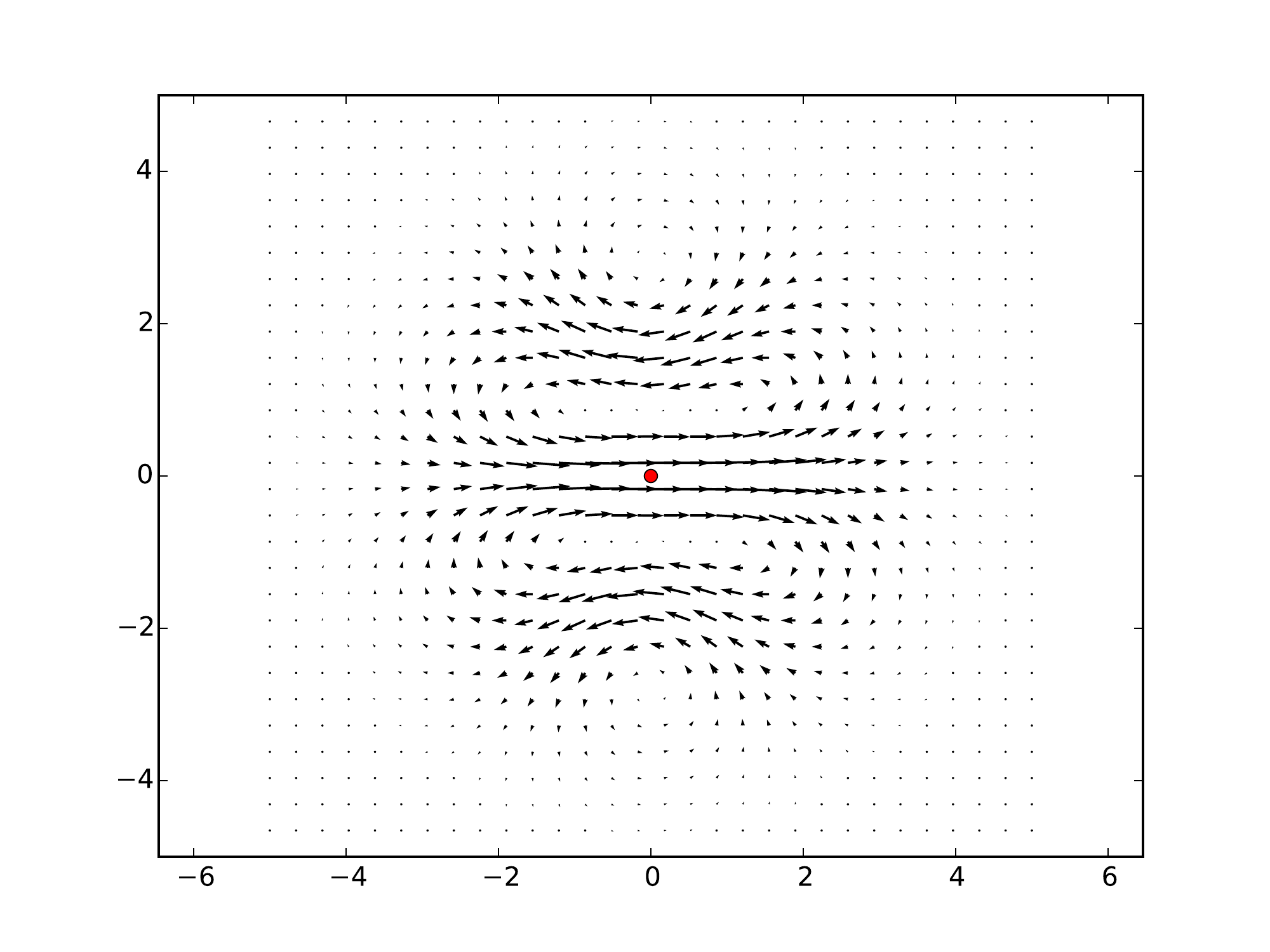}
	\caption{Some velocity fields generated by particles of order $k=2$.}
	\label{fig:2_jet}
  \end{figure}
  
  \subsection{Kelvin's circulation theorem}
  \label{sec:Kelvin}
  In this section we relate the conserved quantities associated
  with $J_R^{(k)}$ to Kelvin's circulation theorem.
Let $\varphi_t \in \SDiff(\R^n)$ denote the flow-map
produced by a solution to Euler's equation.
Let $\gamma(s,t) = \varphi_t(\gamma_0(s))$ for some loop $\gamma_0(s)$.
Kelvin's circulation theorem states that the circulation
\begin{equation*}\label{eq:Kelvin-circulation}
  \Gamma(t) = \oint u( \gamma(s,t) )  \cdot \partial_s \gamma(s,t) ds
\end{equation*}
is constant in time.
It was shown in \cite{Arnold1966} that this conservation law
is an instance of Noether's theorem.
In particular, circulation is one of the conserved momenta
associated with the particle relabeling symmetry of fluids.
More specifically,
recall that $h_0$ is the Hamiltonian for Euler's equations,
and this Hamiltonian is $\SDiff(\R^n)$ invariant. Moreover, we have a
weak\footnote{%
  With $\SDiff(\R^n)$ as right action this would be a proper dual
  pair, but $\iso(z) \subset \SDiff(\R^n)$ does not act transitively
  on the level sets of $J_\text{spatial}$, hence the dual pair is weak.
} dual pair
$J_\text{spatial}: T^*\SDiff(\R^n) \to \mathfrak{X}_{\rm div}(\R^n)^*$
and $J_\text{conv}: T^*\SDiff(\R^n) \to \mathfrak{iso}(z)^*$ of
spatial and convective momentum maps induced by
the left action of $\SDiff(\R^n)$ and the right action of
$\iso(z) \subset \SDiff(\R^n)$ on $\SDiff(\R^n)$ itself.
By applying Theorem~\ref{thm:dual_pairs} with $J_1 = J_\text{spatial}$
and $J_2 = J_\text{conv}$ we know that Hamiltonian dynamics on
$T^*\SDiff(\R^n)$ with respect to a Hamiltonian of the form
$h \circ J_1$ will exhibit the constant of motion
\begin{equation}\label{eq:abstract-circulation}
  \langle J_\text{conv}(\varphi,p_\varphi) , w \rangle
  = \langle p_\varphi , T\varphi \cdot w \rangle
  \qquad \text{for all } w \in \mathfrak{iso}(z).
\end{equation}
This instance of Noether's theorem applies to regularized models as well as to the non-regularized case of Euler's fluid equations.
In the case of Euler's fluid equations, the equivalence between the above conservation law
and Kelvin's circulation theorem is demonstrated by a heuristic
argument~\cite[Chapter~1, Theorem~5.5]{ArnoldKhesin1998}, which goes as follows. 
Consider a given closed curve $\gamma_0\colon S^1 \to \R^n$ and a family
of vector fields $w_\varepsilon$ such that $w_\varepsilon(\gamma_0(s)) = \gamma'_0(s)$
and such that the (weak in $L^2$) limit $w = \lim_{\varepsilon \to 0} w_\varepsilon$ is the generalized function
\begin{equation*}
  w(x) =  \gamma'_0(s) \otimes \delta(x-\gamma_0(s)).
\end{equation*}
Assuming $\varepsilon$ is small and writing
$p_\varphi = \langle \dot{\varphi} , \,\cdot\, \rangle_{L^2}$, we find
\begin{align*}
  \langle J_\text{conv}(\varphi,p_\varphi) , w_\varepsilon \rangle
  &= \langle \dot{\varphi} , T\varphi \cdot w_\varepsilon \rangle_{L^2}\\
  &= \int_{\mathbb{R}^n} \langle \dot{\varphi}(x) , \textrm{D}\varphi(x)\cdot w_\varepsilon(x) \rangle_{\mathbb{R}^n} \,\textrm{d} x \\
  &\approx \oint \langle \dot{\varphi}(\gamma_0(s)) , \textrm{D}\varphi(\gamma_0(s))\cdot \gamma_0'(s) \rangle_{\mathbb{R}^n} \,\textrm{d} s \\
  &= \oint \langle u(\gamma(s,t)) , \partial_s \gamma(s,t) \rangle_{\mathbb{R}^n} \,\textrm{d} s,
\end{align*}
where $u = \dot{\varphi}\circ\varphi^{-1}$ is the usual Euler
representation of the fluid flow and $\gamma(s,t) := \varphi_t(\gamma_0(s))$. 
Therefore, conservation of $J_{\rm conv}$ leads to conservation of circulation.

A more rigorous correspondence is developed in \cite{HolmMarsdenRatiu1998}.
Any $m \in \mathfrak{X}_{\rm div}(\R^n)^*$ can be written as a one-form density $\tilde{m} \otimes \mu$
where $\mu$ is the volume form on $\R^n$.
For any smooth curve $\gamma:S^1 \to \R^n$ we may consider the current
$\mathscr{K}(\gamma) \in \mathfrak{X}_{\rm div}(\R^n)^{**}$ defined by
\begin{equation*}
	\langle \mathscr{K}(\gamma) , m \rangle = \int_{\gamma} \tilde{m}.
\end{equation*}
By Theorem 6.2 of \cite{HolmMarsdenRatiu1998}, if $\gamma_t(s) = \varphi_t(\gamma_0(s))$ and $m = \tilde{m} \otimes \mu$ satisfies
the ideal incompressible fluid equation (perhaps regularized), then the circulation
\begin{equation*}
	\Gamma(t) = \langle \mathscr{K}(\gamma_t) , m(t) \rangle
\end{equation*}
is constant in time.
It is in this sense that Kelvin's circulation theorem follows from the particle
relabeling symmetry for any $\SDiff(\R^n)$ invariant Lagrangian. 

In light of this discussion, it is natural to consider \eqref{eq:abstract-circulation} as the  fundamental conservation law. The main goal in the remainder of this section is to show that the jet-particle solutions satisfy conservation laws that are `shadows' of this fundamental law in the sense that they are associated with a \emph{partial} relabeling symmetry. To see this, it is useful to define the $k$-th order isotropy group,
$\iso^{(k)}(z) = \{ \psi \in \SDiff(\R^n) \mid \Jet^{k}_z( \psi) = \Jet^k_z (\rm{id}) \}$, where we wrote $\rm{id}$ for the identity mapping and $z$ is shorthand for $z_1, \ldots, z_N \in \mathbb{R}^n$.
Then we see that $Q^{(k)}_N = \SDiff(\R^n) / \iso^{(k)}(z)$, and the corresponding quotient map is given by $\Jet_z^k$.
 Let us also introduce the right action $R\colon \SDiff(\R^n) \times \iso(z) \to \SDiff(\R^n)$ given by right composition of functions and write $R_\psi = R(\cdot , \psi)$, so that $R_\psi \varphi = \varphi \circ \psi$. We also introduce the cotangent lifted right action $TR^*$ by means of the defining relation
 \begin{equation}\label{eq:cotangent-action}
 	\langle TR^*_{\psi^{-1}}p_\varphi, v_\varphi \circ \psi \rangle = \langle p_\varphi, v_\varphi\rangle
 \end{equation}
 for arbitrary $p_\varphi \in T_\varphi^*\SDiff(\R^n)$ and $v_\varphi \in T_\varphi\SDiff(\R^n)$.

 The space $T^*Q^{(k)}_N$ naturally embeds into $T^*\SDiff(\R^n) / \iso^{(k)}(z)$, where the quotient is by the cotangent lifted action. More precisely, for any $(q,p) \in T^*Q^{(k)}_N$ we can construct the corresponding element $i(q,p) \in T^*\SDiff(\R^n) / \iso^{(k)}(z)$ in the following way: take any $\varphi$ such that $\Jet_z^k(\varphi) = q$ and find $\iota_\varphi(q,p) \in T_\varphi^*\SDiff(\R^n)$ that satisfies, for all $v_\varphi \in T_\varphi \SDiff(\mathbb{R}^n)$,
 \begin{equation}\label{eq:iota-map}
   \langle \iota_\varphi(q,p) , v_\varphi \rangle = \langle (q, p) , (v_\varphi \circ \varphi^{-1}) \cdot q \rangle,
 \end{equation}
 cf.~\cite[Equation~(2.2.4)]{HRS}, where we denote by $(v_\varphi \circ \varphi^{-1}) \cdot q$ the infinitesimal action from the left of $v_\varphi \circ \varphi^{-1} \in \mathfrak{X}_{\rm div}(\R^n)$ on $q$.
 Then set $i(q,p) = [\iota_\varphi(q,p)]$.

 To see that $i(q,p)$ is well defined, note that if $\varphi' = \varphi \circ \psi$ for $\psi \in \iso^{(k)}(z)$, then for any $v_{\varphi'} \in T_{\varphi'} \SDiff(\R^n)$ we have
 \begin{multline*}
   \langle TR^*_{\psi^{-1}} \iota_\varphi(q,p), v_{\varphi'} \rangle
 = \langle \iota_\varphi(q,p), v_{\varphi'} \circ \psi^{-1} \rangle\\
 = \langle (q, p), (v_{\varphi'} \circ \psi^{-1} \circ \varphi^{-1}) \cdot q \rangle
 = \langle (q, p), (v_{\varphi'} \circ {\varphi'}^{-1}) \cdot q \rangle
 = \langle \iota_{\varphi'}(q,p), v_{\varphi'} \rangle
 \end{multline*}
 using~\eqref{eq:cotangent-action} and~\eqref{eq:iota-map} in the first and second equalities.
 Since $[TR^*_{\psi^{-1}} \iota_\varphi(q,p)] = [\iota_{\varphi'}(q,p)]$, we conclude that $i(q,p)$ is well defined.

 We claim that
 \begin{equation}
 	i(q, p) = [TR^*_{\varphi^{-1}} J_L^{(k)}(q, p)], \label{i_explicit}
 \end{equation}
 where $\varphi$ is such that $\Jet_z^k(\varphi) = q$. This follows since for any $v_{\varphi} \in T_{\varphi}\SDiff(\R^n)$
 \begin{equation*}
	\langle TR^*_{\varphi^{-1}} J_L^{(k)}(q, p), v_{\varphi}\rangle = \langle  J_L^{(k)}(q, p), v_\varphi \circ \varphi^{-1}\rangle = \langle (q, p), (v_\varphi \circ \varphi^{-1}) \cdot q \rangle.
 \end{equation*}

With these preliminary remarks in mind, we can now show the commutativity of the following diagram:
\begin{displaymath}
\xymatrix@R=3em@C=5em{
& \ar@/_1.5pc/[dl]_{J_{\rm spatial}} S \subset T^*\SDiff(\R^n) \ar[d]^{\pi} \ar@/^1.5pc/[dr]^{ \left. J_{\rm conv} \right|_S} & \\
\mathfrak{X}_{\rm div}(\R^n)^* & \ar[l]_-{[J_{\rm spatial}]} i(T^*Q^{(k)}_N) \subset T^*\SDiff(\R^n) / \iso^{(k)}(z) \ar[r]^-{[\left. J_{\rm conv}\right|_{S} ]} & \mathfrak{iso}(z)^* \\
& T^*Q^{(k)}_N \ar[u]^i \ar@/_1.5pc/[ur]_{J_R^{(k)}} \ar@/^1.5pc/[ul]^{J_L^{(k)}} &
}
\end{displaymath}
Here, $S = (\pi^{-1} \circ i)\bigl(T^*Q^{(k)}_N\bigr)$ and $J_{\rm spatial}$ and $J_{\rm conv}$ are defined in the natural manner through the left and right actions of $\SDiff(\R^n)$ and $\iso(z)$.
To verify the left side of the diagram, we let $(q,p) \in T^*Q^{(k)}_N$ and $\varphi$ such that $\Jet_z^k(\varphi) = q$. Then we obtain from \eqref{i_explicit} that for any $w \in \mathfrak{X}_{\rm div}(\R^n)$
\begin{align*}
	&\langle [J_\text{spatial}](i(q,p)), w\rangle = \langle J_\text{spatial}( TR^*_{\varphi^{-1}} J_L^{(k)}(q, p)), w\rangle \\
	\quad&= \langle TR^*_{\varphi^{-1}} J_L^{(k)}(q, p), w \circ \varphi \rangle
	 = \langle J_L^{(k)}(q, p), w \rangle,
\end{align*}
as required. This also shows that $[J_{\rm spatial}]$ is well defined
on $T^*\SDiff(\R^n) / \iso^{(k)}(z)$ since the final expression does
not depend on the choice of representative $\varphi$ for $q$.

For the right leg, note that 
\begin{align*}
	&\langle [J_\text{conv}](i(q,p)), w\rangle = \langle J_\text{conv}( TR^*_{\varphi^{-1}} J_L^{(k)}(q, p)), w\rangle \\
	&= \langle  TR^*_{\varphi^{-1}} J_L^{(k)}(q, p), T\varphi \cdot w \rangle = \langle  J_L^{(k)}(q, p), (T\varphi \cdot w) \circ \varphi^{-1} \rangle \\
	&= \langle (q,p),  \Jet_z^k(T\varphi \cdot w)\rangle = \langle  J^{(k)}_R(q,p), w\rangle.
\end{align*}
Here we used~\eqref{i_explicit}, the definition of $J_\text{conv}$,
and~\eqref{eq:cotangent-action}, respectively, in the first three equalities,
and~\eqref{Moma_right} in the final step. Again we see that $[J_\text{conv}]$ is well-defined since it does not
depend on $\varphi$. More explicitly, for another representative
$\varphi' = \varphi \circ \psi$ with $\psi \in \iso^{(k)}(z)$ we see
that $\psi$ gets projected out by (the tangent map of) $\pi = \Jet_z^k$.
That is, the restricted momentum map $J_\text{conv}|_S$ is invariant under the right action of $\iso^{(k)}(z)$
and therefore descends to a map $[J_\text{conv}|_S ]$ defined on $i\bigl(T^*Q^{(k)}_N\bigr)$.

The right legs of both weak dual pairs yield the conserved quantities.
From the diagram it is clear that the conservation of $J_R^{(k)}$ exhibited in
our particle models is a shadow of the conservation of $J_\text{conv}$ in~\eqref{eq:abstract-circulation}.
Since $J_\text{conv}$ corresponds by Noether's theorem to the (large)
subgroup $\iso(z)$ of the right symmetry $\SDiff(\R^n)$ that generates
conservation of circulation, we see that conservation of $J_R^{(k)}$ in the jetlet solutions is a shadow
of the conservation of circulation.
In other words, our particle models contain a model of Kelvin's circulation theorem within them (cf.~\cite[Theorem 5.5]{JacobsRatiuDesbrun2013}).

This diagram also provides some insight into the relationship between the developments in this paper and classical Marsden-Weinstein reduction theory \cite{MarsdenWeinstein1974}.
For instance, letting $J_{\rm conv}^{(k)}: T^*\SDiff(\R^n) \to \mathfrak{iso}^{(k)}(z)^*$ be the momentum map associated to the cotangent lift of the right action of $\iso^{(k)}(z)$ on $\SDiff(\R^n)$,
 one can show that $i\bigl(T^*Q^{(k)}_N\bigr) = (J_{\rm conv}^{(k)})^{-1}(0) / \iso^{(k)}(z)$.  For more details on the connections between general reduction theory and the results of the present paper, see Appendix \ref{app:diagramatic}.

\section{Particle mergers}
\label{sec:collisions}

In this section we discuss some explicit dynamical behavior of the
particle model, in particular we study `collisions' of jetlets. For
the zeroth order particles case, this was already analyzed by Mumford
and Michor~\cite{MumfordMichor2013}, and they found that two
particles can merge in infinite time, or bounce off each other,
depending on the ratio of their relative angular and linear momenta.
To find the explicit behavior analytically, we shall restrict to two
dimensional space and two $0$-jetlet particles with zero total linear
momentum. We identify the asymptotics of the merged state as the
dynamics of a single $1$-jetlet particle.

We start with the Hamiltonian $H^{(0)}$ for two $0$-jetlet particles,
\begin{equation*}\label{eq:H-two-0jetlets}
  H = \frac{1}{2} \sum_{a,b=1}^2 p\indices{_a_i} K^{ij}(q_a - q_b) p\indices{_b_j}
\end{equation*}
on $T^* \R^{2n}$ with the canonical Poisson brackets.
Translation symmetry allows us perform symplectic reduction. We switch
to a center of mass frame by choosing new coordinates
\begin{equation}\label{eq:CoM-coords}
  \bar{q} = \frac{1}{2}(q_1 + q_2), \qquad \tilde{q} = q_2 - q_1
\end{equation}
with canonically associated momenta
\begin{equation*}
  \bar{p} = p_1 + p_2, \qquad \tilde{p} = \frac{1}{2}(p_2 - p_1).
\end{equation*}
The Hamiltonian in these coordinates becomes
\begin{equation*}
  H = \frac{1}{4} \bar{p}_i \big( K^{ij}(0) + K^{ij}(\tilde{q}) \big)   \bar{p}_j
              + \tilde{p}_i \big( K^{ij}(0) - K^{ij}(\tilde{q}) \big) \tilde{p}_j.
\end{equation*}
When $n = 2$ and the total momentum is zero, i.e.~$\bar{p} = 0$, we can perform
another symplectic reduction by rotational symmetry. We switch to
polar coordinates
\begin{equation}\label{eq:polar-coords}
  \tilde{q} = \big(r\cos(\phi),r\sin(\phi)\big)
\end{equation}
with canonically associated momenta
\begin{equation*}
  \tilde{p}^T =
  \begin{pmatrix}
    \cos(\phi) & -\frac{\sin(\phi)}{r} \\
    \sin(\phi) &  \frac{\cos(\phi)}{r}
  \end{pmatrix} \cdot
  \begin{pmatrix}
    p_r \\
    p_\phi
  \end{pmatrix}
  = R_\phi \cdot
  \begin{pmatrix}
    p_r \\
    \frac{p_\phi}{r}
  \end{pmatrix},
\end{equation*}
where $R_\phi$ is a rotation matrix. In these coordinates the
Hamiltonian is given by
\begin{equation*}
  \begin{aligned}
  H &= \tilde{p}_i \big[ K^{ij}(0) - K^{ij}(\tilde{q}) \big] \tilde{p}_j
     = \begin{pmatrix} p_r \\ \frac{p_\phi}{r} \end{pmatrix}^T R_\phi^T
       \big[ K(0) - K(R_\phi\cdot(r,0)) \big] R_\phi
       \begin{pmatrix} p_r \\ \frac{p_\phi}{r} \end{pmatrix}\\
    &= \begin{pmatrix} p_r \\ \frac{p_\phi}{r} \end{pmatrix}^T
       \big[ K(0) - K(r,0) \big] \begin{pmatrix} p_r \\ \frac{p_\phi}{r} \end{pmatrix},
  \end{aligned}
\end{equation*}
where in the last step we used that $K$ as a tensor is invariant under
rotations. Since $\phi$ is a cyclic variable, we find that its
associated momentum
\begin{equation*}
  p_\phi = -r \sin(\phi) \tilde{p}_1 + r \cos(\phi) \tilde{p}_2
            = \tilde{q} \wedge \tilde{p}
\end{equation*}
is conserved. Remark that this relative angular momentum is not
exactly the total angular momentum when $\bar{p} \neq 0$.

Let us now choose the smooth kernel $K = K_{\infty,1}$ given
(up to a scaling factor) by
{\newcommand{\e}{e^{-\rho}}
\begin{equation*}\label{eq:smooth-kernel}
  K^{ij}(x) = \Big(\e - \frac{1}{2\rho}\big(1 - \e\big)\Big)\delta^{ij}
    + \Big(\frac{1}{\rho}\big(1 - \e\big) - \e\Big)\frac{x^i x^j}{\|x\|^2},
\end{equation*}%
}%
where $\rho := \frac{\|x\|^2}{2}$. Note that
$K^{ij}(0) = \frac{1}{2} \delta^{ij}$ and $\partial_k K^{ij}(0) = 0$.
Using rotational symmetry we set $\phi = 0$, and obtain
\begin{equation}\label{eq:H-rot-reduced}
  H = \frac{p_r^2      }{2}    \Big(1 - \frac{1}{\rho}\big(1-e^{-\rho}\big)\Big)
     +\frac{p_\phi^2}{4\rho}\Big(1 - 2e^{-\rho} + \frac{1}{\rho}\big(1-e^{-\rho}\big)\Big).
\end{equation}
This system is Hamiltonian in $(r,p_r)$ with $p_\phi$ a parameter, so
the level sets of $H$ determine the motion. For $p_\phi = 0$ we see
that $H \propto p_r^2\, r^2 f(r^2)$ where $f$ is an analytic function
with $f(0) > 0$. Thus the level sets of $H$ near $r = 0$ look like
hyperbola with $p_r = 0$ and $r = 0$ as asymptotic axes, see the left image in
Figure~\eqref{fig:contour-red-Hamiltonian}. Hence, two particles approaching
each other head-on will `collide' in infinite time; even though their
momentum blows up, their relative velocity decays exponentially.

If $p_\phi \neq 0$ then the first non-constant contributing term has
sign opposite to the terms involving $p_r$, hence we see a new region
being created in the right image of Figure~\eqref{fig:contour-red-Hamiltonian} where orbits
approach $r = 0$ and then `bounce off'. There are two asymptotic
orbits that approach the $r = 0$ axis at finite momentum; the limit
point can be calculated from the fact that $\dot{p}_r / \dot{r} = 0$
must hold there in the limit $r \to 0$. We find that
$p_r = \pm \sqrt{5/6}\,p_\phi$. From
Figure~\ref{fig:contour-red-Hamiltonian} it is clear that this is a good
approximation for the asymptotic curve.

\begin{figure}[ht]
  \centering
  \includegraphics[width=7.5cm]{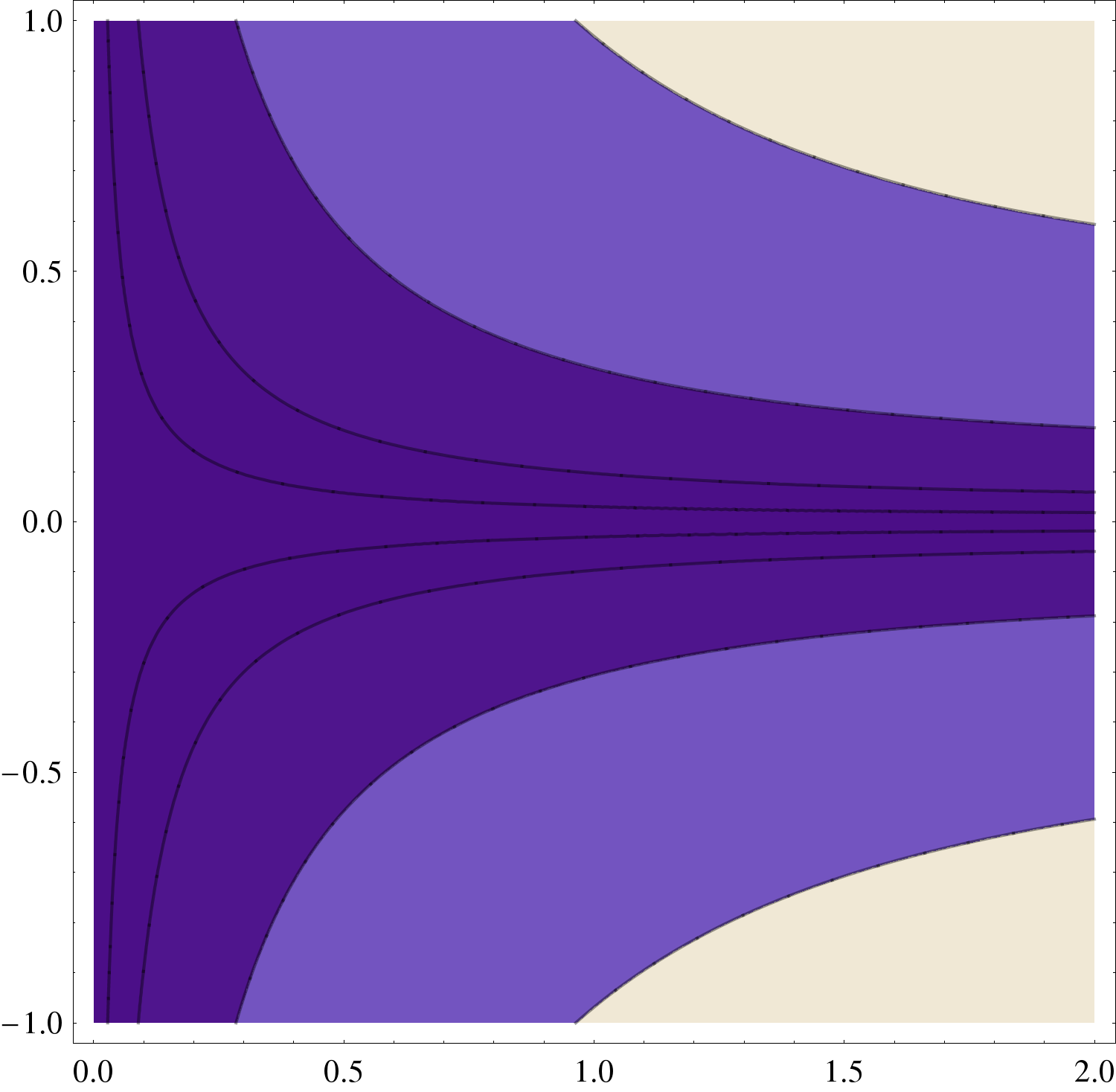}
  \hspace{0.5cm}
  \includegraphics[width=7.5cm]{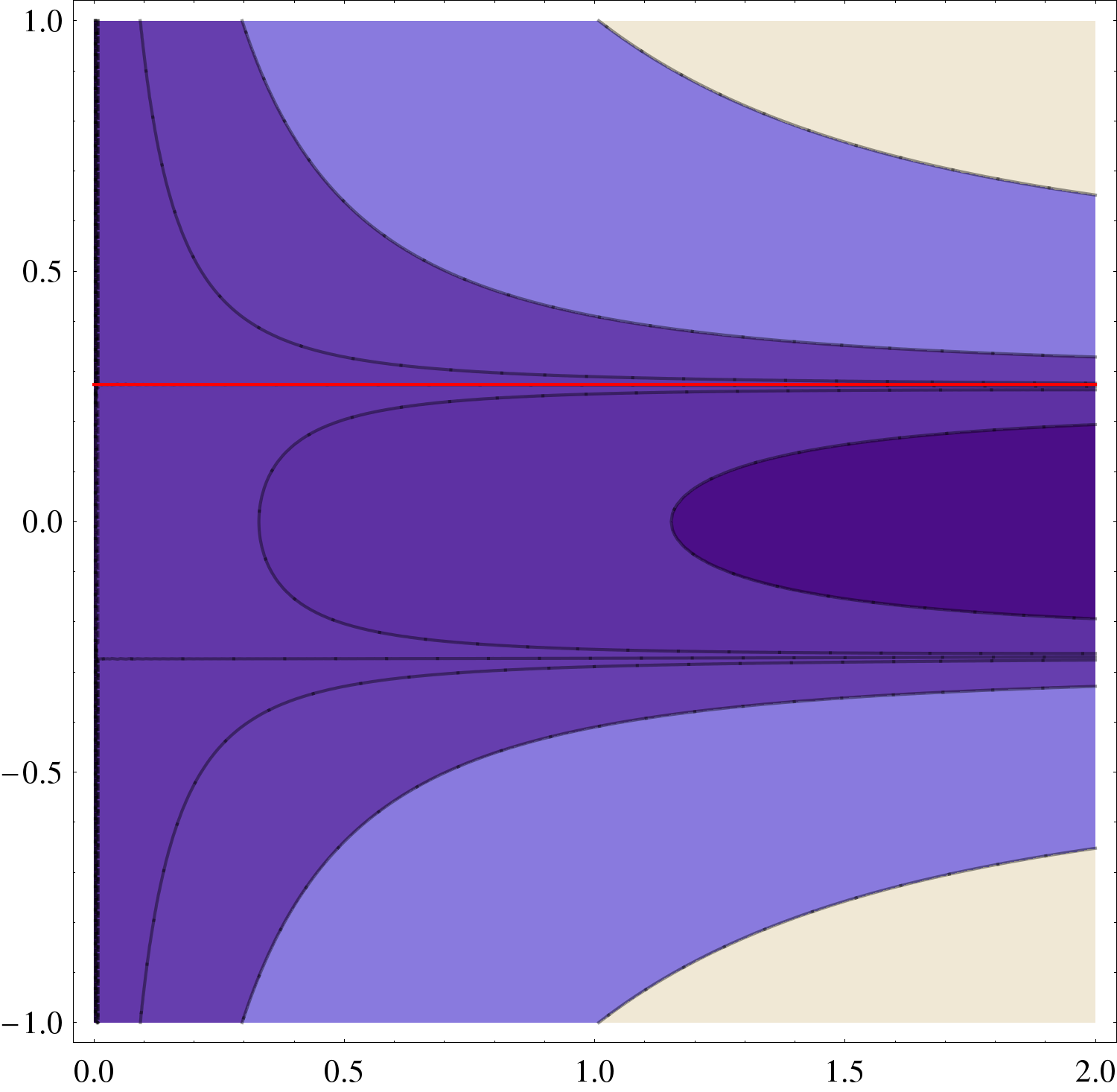}
  \caption{Contour plots of the reduced
    Hamiltonian~\eqref{eq:H-rot-reduced} with $r$ horizontal and $p_r$
    vertical, left for $p_\phi = 0$ and right for
    $p_\phi = 0.3$. The red line shows the positive asymptotic
    value for merging or bouncing.}
  \label{fig:contour-red-Hamiltonian}
\end{figure}

Finally, we can reconstruct the asymptotic $1$-jetlet trajectory that
these two merging $0$-jetlets converge to and verify that this
trajectory is indeed a solution of the Hamiltonian vector field for a
$1$-jetlet.

To analyze the asymptotic behavior, we expand $H$ around $\rho = 0$:
\begin{equation*}
  H = \frac{p_r^2      }{4} \rho
     +\frac{p_\phi^2}{8}\big(3 - \frac{5}{3}\rho\big)
     +\mathcal{O}(\rho^2).
\end{equation*}
Since $H$ and $p_\phi^2$ are preserved, we can solve for $p_r$ in
terms of $r$, and we find
\begin{align*}
  2\rho\,p_r^2 &= (8H - 3 p_\phi^2) \big(1 + \mathcal{O}(\rho)\big), \\
\Longleftrightarrow \qquad
      r\,p_r   &= \sqrt{8H - 3 p_\phi^2} + \mathcal{O}(r).
\end{align*}
We write $\zeta := \sqrt{8H - 3 p_\phi^2}$ and thus obtain
asymptotically $p_r = \frac{\zeta}{r} + \mathcal{O}(1)$. Further, we
have dynamics
\begin{equation*}
  \dot{r} = \frac{\partial H}{\partial p_r} = \frac{1}{4}p_r\,r^2 + \mathcal{O}(r^4), \qquad
  \dot{\phi} = \frac{\partial H}{\partial p_\phi}
                 = \frac{3}{4}p_\phi + \mathcal{O}(r^2)
\end{equation*}
and reconstruct
\begin{alignat*}{2}
  \bar{q}   &= 0,& \qquad
  \tilde{q} &= r R_\phi \cdot\begin{pmatrix} 1 \\ 0 \end{pmatrix},\\
  \bar{p}   &= 0,& \qquad
\tilde{p}^T &= R_\phi \cdot
               \begin{pmatrix} p_r \\ \frac{p_\phi}{r} \end{pmatrix}
             = \frac{1}{r} R_\phi \cdot
               \begin{pmatrix} \zeta \\ p_\phi \end{pmatrix} + \mathcal{O}(1).
\end{alignat*}
Now we consider the image under $J_L$ of the asymptotic solution
curve. By testing against a vector field in $\mathfrak{X}_{\rm
  div}(\R^n)$ we find
\begin{align}
  J_L(q_1,p_1,q_2,p_2)
  &= p_1 \otimes \delta_{q_1} + p_2 \otimes \delta_{q_2} \notag\\
  &= \bar{p} \otimes \delta_{\bar{q}} +\tilde{p} \otimes
     \big({-\tilde{q}^j} \partial_j\delta_{\bar{q}}\big)
    +\mathcal{O}\Big(\big(|\bar{p}|+|\tilde{p}|\big)|\tilde{q}|^2\Big) \notag\\
\intertext{and inserting our reconstructed solution leads to}
  &= -\Big[R_\phi \cdot \begin{pmatrix} \zeta \\ p_\phi \end{pmatrix}\Big]_i
      \Big[R_\phi \cdot \begin{pmatrix} 1   \\ 0         \end{pmatrix}\Big]^j \dx^i\,\partial_j
      \delta_{\bar{q}} +\mathcal{O}(r) \label{eq:asymptotic-dynamics}
\end{align}
with $\phi(t) = \frac{3}{4} p_\phi\,t$. Our aim is to show that
the factor in front of $\dx^i\,\partial_j \delta$, which we shall denote
by $\tilde{\mu}\indices{_i^j}(t)$, corresponds to
$\mu\indices{_i^j}(t)$ for a $1$-jetlet with position $q^{(0)} = 0$
and momentum $p^{(0)} = 0$. From here on, we use the Frobenius inner
product to identify $\mathfrak{sl}(2)^* \cong \mathfrak{sl}(2)$, which
for explicit matrices corresponds to taking the transpose.
For such a setup we have equations of motion
\begin{equation*}
  \dot{\mu}\indices{_i^j}
  = \mu\indices{_i^k} \partial_k u^j(0) - \partial_i u^k(0) \mu\indices{_k^j}
  = -\partial_{km} K^{jl}(0) \mu\indices{_l^m} \mu\indices{_i^k}
    +\partial_{im} K^{kl}(0) \mu\indices{_l^m} \mu\indices{_k^j},
\end{equation*}
or in short, $\dot{\mu} = -\ad^*_\xi(\mu) = [\mu,\xi^T]$ with
$\xi\indices{^i_j} = \partial_j u^i(0) = -\partial_{jk} K^{i l}(0) \mu\indices{_l^k}$.
To ease calculations, let us choose the basis
\begin{equation*}\label{eq:basis-sl2}
  \omega = \begin{pmatrix}
    0 & -1 \\
    1 & 0
  \end{pmatrix},\qquad
  \sigma = \begin{pmatrix}
    1 & 0 \\
    0 & -1
  \end{pmatrix},\qquad
  \tau = \begin{pmatrix}
    0 & 1 \\
    1 & 0
  \end{pmatrix}
\end{equation*}
and note that these matrices have norm $\sqrt{2}$. We decompose
$\mu = j\,\omega + s\,\sigma + t\,\tau$ and the second derivative of
the kernel as a tensor product. A calculation verified by symbolic
computer algebra software shows that
\begin{equation*}
  \partial_{ij} K^{kl}(0)
  = \frac{1}{4} \left(\begin{array}{rr}
     -\begin{pmatrix} 1 & 0 \\ 0 & 3 \end{pmatrix} &
      \begin{pmatrix} 0 & 1 \\ 1 & 0 \end{pmatrix} \\[2.5ex]
      \begin{pmatrix} 0 & 1 \\ 1 & 0 \end{pmatrix} &
     -\begin{pmatrix} 3 & 0 \\ 0 & 1 \end{pmatrix}
    \end{array}\right)
  =-\frac{1}{2} \omega_i^k \omega_j^l
   -\frac{1}{4} \big( \sigma_i^k \sigma_j^l + \tau_i^k \tau_j^l \big),
\end{equation*}
where the indices $i,j$ and $k,l$ label the outer and inner matrix
elements respectively. With these decompositions we find that $\xi$ as
a function of $\mu$ can be written as
\begin{equation*}
  \xi = j\,\omega + \frac{1}{2}(s\,\sigma + t\,\tau).
\end{equation*}
Using the commutation relations $[\omega,\sigma] = 2\tau$,
$[\omega,\tau] = -2\sigma$, $[\sigma,\tau] = -2\omega$, it then
follows that
\begin{equation}\label{eq:mu-dynamics}
  \dot{\mu}
  = [\mu , \xi^T]
  = \big[j\,\omega+s\,\sigma+t\,\tau \,,\, -j\,\omega + \frac{1}{2}(s\,\sigma + t\,\tau)\big]
  = -3j\big(-t\,\sigma + s\,\tau\big)
  = -\frac{3}{2}j\,[\omega,\mu].
\end{equation}
That is, the $(\sigma,\tau)$ components of $\mu$ as a tensor rotate
with angular velocity $-3j$.

On the other hand, from~\eqref{eq:asymptotic-dynamics} we have for the
asymptotic dynamics that
\begin{equation*}
  \tilde{\mu}(t)
  = -\Big[R_\phi \cdot \begin{pmatrix} \zeta \\ p_\phi \end{pmatrix}\Big] \cdot
     \Big[R_\phi \cdot \begin{pmatrix} 1   \\ 0         \end{pmatrix}\Big]^T
  = -\Ad_{R_\phi} \begin{pmatrix} \zeta & 0 \\ p_\phi & 0 \end{pmatrix}
\end{equation*}
with only $\dot{\phi} = \frac{3}{4} p_\phi$ depending on time.
Note that $\tilde{\mu} \not\in \mathfrak{sl}(2)$ as a matrix, but
since it is actually a dual element, we can simply ignore its trace
part and project it out. Also note that the $\omega$ component of
$\tilde{\mu}(0)$ is $j = \frac{1}{2}p_\phi$. Differentiating with respect to 
time yields
\begin{equation*}
  \dot{\tilde{\mu}}
  = -\dot{\phi}\,\ad_\omega(\tilde{\mu})
  = -\frac{3}{4} p_\phi [\omega,\tilde{\mu}]
  = -\frac{3}{2} j [\omega,\tilde{\mu}], 
\end{equation*}
and comparing to~\eqref{eq:mu-dynamics} we find that the asymptotic
solution of the two merging $0$-jetlets matches that of a $1$-jetlet
with the same angular (and linear) momentum.

Let us finally suggest a more abstract way to view these particle
mergers. Our hierarchy of reduced spaces $T^*Q_N^{(k)}$ embeds into
$\mathfrak{X}_{\rm div}(\R^n)^*$ under the momentum map $J_L$.
Consider a merging pair of $0$-jetlets, described by a curve
$x_0(t) \in T^*Q_2^{(0)}$. As the particles approach each other,
$x_0(t)$ approaches the boundary of $T^*Q_2^{(0)}$ given by
\begin{equation*}\label{eq:boundary}
  \partial\big( T^*Q_2^{(0)} \big)
  = \{ (q_1,p_1,q_2,p_2) \in T^* \R^n \times T^* \R^n \mid q_1 = q_2 \}.
\end{equation*}
On the other hand, the image curve
$y_0(t) = J_L(x_0(t)) \in \mathfrak{X}_{\rm div}(\R^n)^*$ consists of
two covector-valued delta distributions at $q_1,q_2$, and in the limit
as their distance goes to zero, this can be approximated by a momentum
valued distribution of a delta and its
derivative~\eqref{eq:asymptotic-dynamics}, that is, an element
$y_1(t) = J_L(x_1(t))$, where $x_1(t) \in T^*Q_1^{(1)}$ is a curve in
the space of single $1$-jetlet particles.

We can view the boundary of $T^*Q_2^{(0)}$ as a subset of $T^*Q_1^{(1)}$,
and consider a topology on $\mathfrak{X}_{\rm div}(\mathbb{R}^n)^*$ in which
the embedding is continuous,
see diagram~\eqref{eq:embedding}.
This picture naturally generalizes to the whole hierarchy of spaces
$T^*Q_N^{(k)}$, suggesting that it might be interpreted as a
CW-complex. In this setting, the question whether the solution curve
$y_0(t)$ of the merging $0$-jetlets converges to a solution curve
$y_1(t)$ of a $1$-jetlet, basically\footnote{%
  One has to be careful, however, since continuity of the vector field will
  only imply that the asymptotic curve $\tilde{y}_0(t) \in J_L(T^*Q_1^{(1)})$
  is a pseudo orbit of the $1$-jetlet dynamics. This does not imply
  existence of a solution curve $y_1(t) \in J_L(T^*Q_1^{(1)})$ that
  $\tilde{y}_0(t)$ is asymptotic to; that would require the `limit
  shadowing property', which is closely related to hyperbolic
  properties of the dynamics~\cite{Palmer2012,Ribeiro2014}.%
} boils down to the question whether the vector field of the dynamics
on $\mathfrak{X}_{\rm div}(\R^n)^*$ is continuous.
We have not pursued in detail the question of which topology on
$\mathfrak{X}_{\rm div}(\R^n)^*$ to use for this more abstract
characterization.


\begin{equation} \label{eq:embedding}
   \xymatrix{
     T^*Q_N^{(k)} \ar[r]^{J_L} & \mathfrak{X}_{\rm div}(\R^n)^* \\
     T^*Q_1^{(1)} \ar[ur] \ar@{.>}[u] \\
     T^*Q_2^{(0)} \ar[uur] \ar[u]^{\partial}
   }
\end{equation}

\section{Numerical experiments}
\label{sec:Numerical experiments}

We have performed a number of numerical simulations\footnote{%
  The simulation was written using Python and NumPy, the source code
  and generated videos can be found at:
  \url{https://github.com/hoj201/incompressible_jet_particles}%
} of jetlet particles. These confirm that the conserved
quantities are indeed preserved and that two $0$-jetlet particles
merge as shown by the analysis in Section~\ref{sec:collisions},
providing a sanity check for the formulas in the previous section and
the numerical code. Moreover, the merging behavior shows to be stable
under perturbations of initial conditions.

The numerical code works for any spatial dimension $n \ge 2$, but for
the sake of tractability and simplicity we have studied $n = 2$ 

As basic experiment we take two $0$-jetlet particles with initial
states
\begin{equation}\label{eq:exp1-initial}
  \begin{alignedat}{2}
    q_1 &= (-3, 0), &\qquad& p_1 = ( 1.5,-d), \\
    q_2 &= ( 3, 0), &      & p_2 = (-1.5, d),
  \end{alignedat}
\end{equation}
aimed at each other with an offset parameter $d$.
The initial state is given in center of mass polar coordinates,
see~\eqref{eq:CoM-coords} and~\eqref{eq:polar-coords}, by
\begin{align*}
    r    &= 6, \qquad p_r    = -1.5, \\
    \phi &= 0, \qquad p_\phi = 6d.
\end{align*}
Furthermore, we use
$\sigma = 1$ throughout our experiments. The experiments show that for
$d = 0.27$ the two particles merge while spinning around each other,
while for $d = 0.288$ they get close, but then emerge from their close
spinning state and scatter in opposite directions. This confirms the
analytical value of $d = \sqrt{3/40} \approx 0.2739$ within reasonable
precision, noting that this is the asymptotic value for particles
starting close to each other.

We performed a number of more complex simulations, all of
those being small perturbations of the basic experiment described
above. First, we added a small angular momentum `spin' component to
both particles, turning both into $1$-jetlets, one level higher in the
hierarchy. Then we added a third particle (both a $0$-jetlet and
$1$-jetlet) at such a distance and momentum that it exhibits medium
range interaction with the first two particles. Finally, we added a
small hyperbolic-like `stretching' momentum to the first particle
only. We found an analytic study of these configurations to be
infeasible, but the simulations show that the behavior observed in the basic
experiment persists. We can find parameter values of $d$ close to the
original one where the system shows a transition between the two
particles merging or scattering.

These experiments also confirm the preservation of the conserved
quantities present in the system. For all experiments described above
we observed that the energy, total linear and angular momentum,
as well as $J_R$ (individually for each particle) were
preserved with absolute errors less than $4 \cdot 10^{-4}$ over a time
of $60$ seconds, while the energy was of the order one. Unlike $J_R$,
which is conserved for each particle, linear and angular momentum can be
exchanged between particles (although total momentum is conserved).
Figure~\ref{fig:angmom-exchange} left shows
the angular momentum of two scattering $1$-jetlets and
the right plot of two scattering $0$-jetlets
interacting weakly with a third jetlet particle.

\begin{figure}[htb]
  \centering
  \includegraphics[width=8cm]{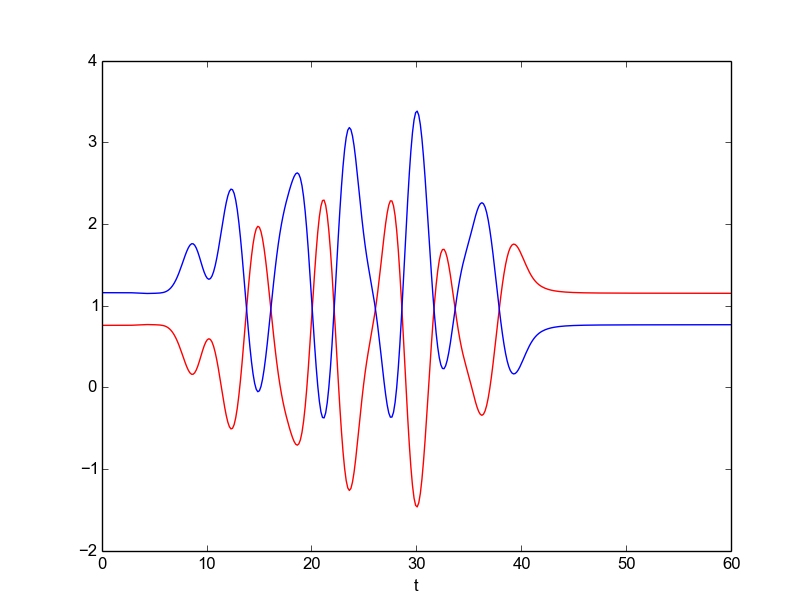}
  \includegraphics[width=8cm]{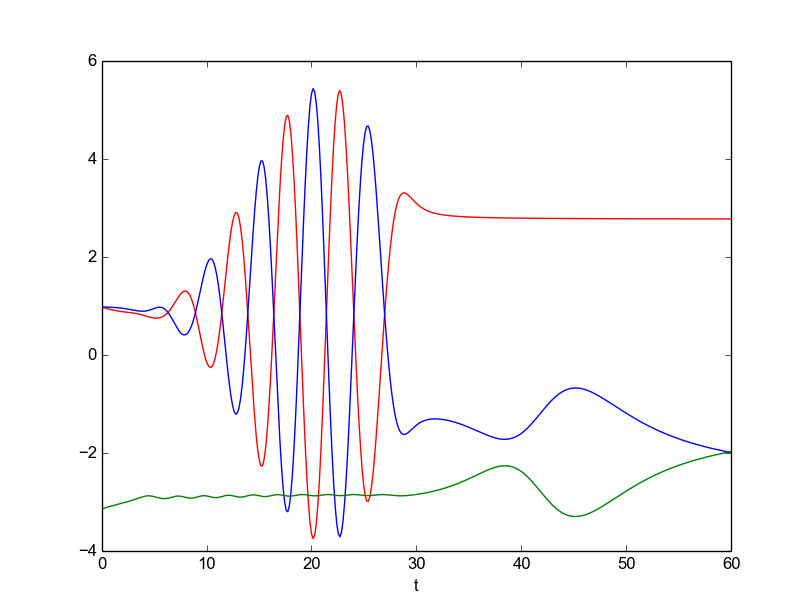}
  \caption{Angular momentum exchange between two scattering
    $1$-jetlets (left) and with a third particle present (right).}
  \label{fig:angmom-exchange}
\end{figure}

\section{Conclusion}
  In this paper we derived a hierarchy of weak dual pairs
  which induce a family of particle-like solutions, called \emph{jetlets}, and conserved quantities
  that shadow the conservation of circulation.
  The jetlets have internal degrees of freedom given by a jet group.
  As the jet group is a finite-dimensional model of the diffeomorphism
  group, we suggested the use of jetlets as a finite-dimensional model
  of self-similarity, wherein a ``large'' diffeomorphism advects a ``small''
  diffeomorphism.
  We also studied the dynamics of mergers and provided a rigorous analysis showing that merging
  0-jetlets asymptotically approach 1-jetlets.

The developments discussed in the present paper give rise to a number of promising directions for future research. These include:
\begin{enumerate}
	\item An investigation of the relationship between jetlets and point vortices or vortex blobs.
	\item Further investigation of the numerical implementation.
		The use of parallelization and the fast multipole method would be particularly interesting to consider.
	\item Finding a way to implement boundary conditions.
		In such scenarios, the kernel is no longer invariant under rigid transformations
		and we must consider a general kernel $K:M \times M \to \mathbb{R}^{n \times n}$
		where $M \subset \mathbb{R}^n$ is an $n$-manifold with boundary.
        \item An analysis of convergence to Euler equations when $\sigma \to 0$ for the case where the power of the Helmholtz operator, $p$, goes to infinity.  The advantage of the $p=\infty$ case is that the limiting kernel can be written in terms of elementary functions \cite{MicheliGlaunes2014}.
\end{enumerate}

\section{Acknowledgements}
We are indebted to the anonymous referees for very carefully refereeing our article, including catching a problem with our initial use of dual pairs.
JE, DDH, HOJ and DMM are grateful for partial support by the European Research Council Advanced Grant 267382 FCCA.

\appendix

\section{Hamiltonian mechanics}
\label{sec:Hamiltonian}
The goal of this section is to prove Theorem~\ref{thm:dual_pairs} (see page~\pageref{thm:dual_pairs}).
Those who understand and accept these theorems on a first reading
should be able to skip this section without any consequence.
Most of this section will be a crash course in Poisson geometry
and Hamiltonian mechanics as described in \cite{FOM,MandS,Weinstein1983}.

A typical introduction to Poisson structures in mechanics
begins by considering Hamilton's equations
\begin{equation*}
  \dot{q} = \pder{H}{p} \quad, \quad  \dot{p} = - \pder{H}{q}.
\end{equation*}
If we consider the two-form $dp \wedge dq$, then Hamilton's equations
can be written as $(\dot{q},\dot{p}) \contract (dp \wedge dq) = dH(q,p)$.
This is the starting point for symplectic geometry, which
will be discussed in Section~\ref{sec:Symplectic}.
Alternatively, these equations can be written using the bilinear map $\{ \cdot , \cdot \}_{\rm can} : C^1( T^*Q) \times C^1(T^*Q) \to C^0(T^*Q)$
given by
\begin{equation*}
  \{ F , G \}_{\rm can} = \pder{F}{q}\pder{G}{p} - \pder{G}{q} \pder{F}{p}.
\end{equation*}
In particular, we may write $\dot{q} = \{ q , H \}_{\rm can}$ and $\dot{p} = \{ p , H\}_{\rm can}$.
The object $\{ \cdot , \cdot \}_{\rm can}$ is a special case of more general
object known as a \emph{Poisson bracket} which will be introduced in
Section~\ref{sec:Poisson}.

A word of warning: symplectic geometry has developed greatly
since its origins in mechanics, and
has branched into an independent subfield of pure mathematics.
Many notions were revised and optimized in the $1970$'s and $1980$'s for
the purpose of proving theorems.
Occasionally these revisions entailed a sacrifice
in clarity, from the perspective of ``outsiders''.
This paper is intended to allow ``outsiders''
(such as ourselves) to reap the benefits of Poisson geometry.
Therefore, we will cut away as much abstraction as possible in this introductory
section.
Nonetheless, a minimal amount of abstraction is needed in order to
maintain mathematical rigor and stand firmly upon the shoulders of giants.

\subsection{Symplectic manifolds}
\label{sec:Symplectic}
We begin with the definition.
\begin{defn}
  Let $S$ be a manifold
  and let $\omega$ be a closed two-form on $S$ such that the map
  ``$v \in TS \mapsto \omega( v , \cdot ) \in T^*S$'' is weakly non-degenerate.\footnote{
    A linear map $L:V \to V^*$ is \emph{weakly non-degenerate} if $L$ is injective.  If $V$ is finite-dimensional, this simply means that $L$ is invertible.}
  We call $\omega$ a \emph{symplectic form}.
  We call the pair $(S,\omega)$ a \emph{symplectic manifold}.
\end{defn}
All of the expressions derived in this article are formal, and we refer to \cite{GayBalmazVizman2012}
for the functional analytic details of infinite-dimensional symplectic manifolds.
As a first example, consider the manifold $\R^2$
with coordinates $(q,p)$.
The two-form $dq \wedge dp$ is a symplectic form.
Given a manifold $Q$,
the cotangent bundle $T^*Q$ has local fiber bundle coordinates
given by $(q^1,\dots, q^n,p_1,\dots,p_n)$ and there is a unique symplectic form
which is locally expressed by $dp^i \wedge dq_i$, where
a sum on repeated indices is assumed.
This local expression corresponds to a global symplectic
form on $T^*Q$, known as the \emph{canonical symplectic form}
and denoted $\omega_{\rm can}$ \cite[Theorem 3.2.10]{FOM}.
In fact, given any symplectic manifold $(S,\omega)$, the dimension of $S$
is even, and there exist local coordinates $(q^1,\dots,q^n,p_1,\dots,p_n)$
such that $\omega \stackrel{locally}{=} dp_i \wedge dq^i$.
This is known as \emph{Darboux's theorem} and we call this type of
coordinates \emph{Darboux coordinates}\cite[Theorem 3.2.2]{FOM}.

Given a function $H:S \to \mathbb{R}$,
the exterior derivative is the one-form $dH:S \to T^*S$
expressed in local coordinates by $dH(x) = \pder{H}{x^i} \dx^i$.
The Hamiltonian vector field $X_H:S \to TS$ is the unique vector
field defined by the condition
$
  X_H \contract \omega = dH.
$
The symbol ``$\contract$'' is the operation of contraction between
the contravariant indices of $X_H$ and the first set of covariant
indices of $\omega$.
In Darboux coordinates, the Hamiltonian vector field induces the
equations of motion $\dot{q}^i = \pder{H}{p_i}$, $\dot{p}_i = -\pder{H}{q^i}$.

An important aspect of study in Hamiltonian mechanics is that of symmetry.
This yields the following notions.

\begin{defn}
  Let $G$ be a Lie group and let $\rho:G \to \Diff(S)$ be a group action on a symplectic manifold $(S,\omega)$.
  The group $G$ is said to \emph{act symplectically} if
  \begin{equation*}
    \omega( \rho( g ) \cdot v , \rho(g) \cdot w) = \omega(v,w)
  \end{equation*}
  for any $g \in G, v,w \in T_xS$ and $x \in S$.
  If $\mathfrak{g}$ is the Lie algebra of such a group,
  the \emph{momentum map}, $J:S \to \mathfrak{g}^*$,
  is defined by the property
  \begin{equation*}
    d \langle J , \xi \rangle = \xi_S \contract \omega.
  \end{equation*}
\end{defn}

Alternatively, we can characterize a momentum map $J: S \to \mathfrak{g}^*$, as the unique map such that $X_{\langle J , \xi \rangle} = \xi_S$ for any $\xi \in \mathfrak{g}$.
In the special case where $S = T^*Q$, a left/right action of $G$ on $Q$ can be lifted to a right/left symplectic action on $T^*Q$ given by
\begin{equation*}
  (q,p) \in T^*Q \mapsto (g^{-1} \cdot q , g^* p) \in T^*Q.
\end{equation*}
where $g^*p$ is the unique covector such that $\langle g^*p , v \rangle = \langle p , Tg \cdot v \rangle$.
In this case the momentum map is characterized by the condition
\begin{equation}
  \langle J(q,p) , \xi \rangle = \langle p , \xi \cdot q \rangle.
  \label{eq:cotangent_momap}
\end{equation}
This is contained in Theorem 12.1.4 of \cite{MandS}.

Finally, given two functions $f,h \in C^{\infty}(S)$ we can consider the function $\{ f,h\} = \omega( X_f, X_h)$.
In Darboux coordinates $\{ f , h \} = \pder{f}{q^i} \pder{h}{p_i} - \pder{f}{p_i} \pder{h}{q^i}$.
Hamilton's equations can then be written as $\dot{q}^i = \{ q^i , h \}$,
$\dot{p}_i = \{ p_i , h \}$.
We call $\{ \cdot , \cdot \}$ a Poisson bracket, and it is the subject of
the next subsection.

\subsection{Poisson manifolds}
\label{sec:Poisson}
We begin with the definition.

\begin{defn} \label{defn:Poisson}
  Let $P$ be a manifold, and $\{ \cdot , \cdot \}$ be a bilinear
  operation on $C^{\infty}(P)$ such that 
  $( C^{\infty}(P) , \{ \cdot , \cdot \} )$ is a Lie algebra
  and $\{ \cdot , h \}$ has the derivation property for any $h \in C^{\infty}(P)$.
  That is to say
  \begin{equation*}
    \{ gf , h \} = \{ f , h \} \cdot g + \{ g , h \} \cdot f,
  \end{equation*}
  for any $f,g,h \in C^{\infty}(P)$.
  We call $\{ \cdot , \cdot \}$ a \emph{Poisson bracket},
  and we call the pair $(P, \{ \cdot , \cdot \})$ a Poisson manifold.
\end{defn}

The most important example of a Poisson bracket is
that of a Poisson bracket on a symplectic manifold $(S,\omega)$.
Here the Poisson bracket is $\{ f , g \} = \omega( X_f , X_g )$.
When $S$ is a cotangent bundle, and $\omega$ is the canonical
symplectic form, we call this bracket the \emph{canonical Poisson
bracket}.

The second most important example of a Poisson bracket,
after the canonical Poisson bracket,
is the Lie--Poisson bracket.  Let $\mathfrak{g}$ be a Lie algebra
and let $\mathfrak{g}^*$ denote its dual.
The \emph{Lie--Poisson bracket} on $\mathfrak{g}^*$ is given 
by
\begin{equation}
  \{ f , g \}_{\mathfrak{g}^*}( x ) = \pm
  \left \langle x , \left[ \pder{f}{x} , \pder{g}{x} \right] \right \rangle, 
  \label{eq:Lie-Poisson}
\end{equation}
where $\langle \cdot , \cdot \rangle$ is the canonical pairing between
dual-vectors and vectors, and $[ \cdot , \cdot ]$ is the Lie bracket
on $\mathfrak{g}$.
The ``$+$'' Poisson bracket is nothing but the canonical Poisson bracket on $T^*G$,
mapped to the space $\mathfrak{g}^*$ via the left trivialization map $\lambda: (g,p) \in T^*G \mapsto (L_g)^*p \in \mathfrak{g}^*$.
The ``$-$'' bracket is obtained through the right trivialization map
$\rho:(g,p) \in T^*G \mapsto (R_g)^*p \in \mathfrak{g}^*$.

On a Poisson manifold $(P,\{ \cdot , \cdot \})$
the derivation property implies that the functional operator
$\{ \cdot , h \}$ is equivalent
to the Lie derivative operator of a unique vector field $X_h:P \to TP$.
That is to say, $X_h$ is the unique vector field such that $\lie_{X_h}[f] = \{ f , h \}$ for any $f \in C^1(P)$.
We call $X_h$ the \emph{Hamiltonian vector field} and the ODE $\dot{x} = X_h(x)$ is called a \emph{Hamiltonian equation}.
It is standard to write this ODE as ``$\dot{x} = \{ x , H\}$'',
despite the fact that one typically intends for ``$x$'' to represent
a point in $P$, and not a function.
Since one can take ``$x$'' to be a place-holder for a set of
local coordinate functions which determine $x$ uniquely, this
sloppiness is usually harmless.

\begin{prop}[Proposition 10.2.2 \cite{MandS}] \label{prop:Lie_hom}
  Let $(P,\{ \cdot , \cdot \})$ be a Poisson manifold.
  Then $X_{ \{ h ,f \} } = - [X_h , X_f ]$.
\end{prop}

\begin{cor} \label{cor:Lie_hom}
  Let $(S,\omega)$ be a symplectic manifold
  and let $h,f \in C^{\infty}(S)$.
  Then $[X_h , X_f] = -X_{\omega(X_h,X_f) }$.
\end{cor}
\begin{proof}
  $X_{\omega(X_h,X_f)} = X_{ \{h,f\} } = -[X_h , X_f]$.
\end{proof}

\begin{defn}
  Let $(P_1, \{ \cdot , \cdot \}_1)$ and $(P_2, \{ \cdot , \cdot \}_2)$
  be Poisson manifolds.
  A map $\psi:P_1 \to P_2$ is called a
  \emph{Poisson map} if $\{ f \circ \psi , g \circ \psi \}_1 = \{ f , g \}_2 \circ \psi$  for any $f,g \in C^{\infty}(P_2)$.
\end{defn}

\begin{prop}[Lemma 1.2 of \cite{Weinstein1983} or Proposition 10.3.2 of \cite{MandS}] \label{prop:Poisson_dynamics}
  Let $\psi:P_1 \to P_2$ be a Poisson map.
  Let $h_2 \in C^1(P_2)$.
  If $x(t) \in P_1$ is a solution to Hamilton's equations with respect
  to $h_1 = h_2 \circ \psi \in C^1(P_1)$, then $y(t) = \psi(x(t)) \in P_2$ is a solution
  to Hamilton's equations with respect to $h_2$.
\end{prop}

  Remark that unlike in~\cite{MandS}, we only require $C^1$~smoothness
  since we do not use existence and uniqueness of solutions.

  When the dimension of $P_2$ is larger than that of $P_1$,
  Proposition~\ref{prop:Poisson_dynamics} allows one to find solutions of
  Hamiltonian equations on $P_2$
  by solving lower-dimensional Hamiltonian equations
  on $P_1$.
  \subsection{Weak dual pairs}
  In this section we review the notion of weak dual pairs
  \cite{GayBalmazVizman2012}.
  This is a relaxation of the more frequently invoked notion of a dual pair \cite{MarsdenWeinstein1983,Weinstein1983}.
  Let $(S,\omega)$ be a symplectic manifold.
  Given a distribution $V \subset TS$, denote the fiber over
  $x \in S$ by $V_x \subset T_x S$.
  The \emph{symplectic orthogonal} to $V$ is the distribution
  \begin{equation*}
    V^\omega = \{ w \in TS \mid \omega( w , v ) = 0, \forall v \in V \}.
  \end{equation*}
  \begin{defn}[Weak dual pair \cite{GayBalmazVizman2012}]\label{def:weak_dual_pair}
  Let $J_1\colon S \to P_1$ be a Poisson map.
  The kernel of $J_1$ is the distribution
  \begin{equation*}
    \kernel(J_1) = \{ v \in TS \mid TJ_1 \cdot v  = 0 \}.
  \end{equation*}
  If $J_2\colon S \to P_2$ is a Poisson map as well, and
  \begin{equation*}
    \kernel(J_2)^\omega \subset \kernel(J_1) \quad , \quad
    \kernel(J_1)^\omega \subset \kernel(J_2)
  \end{equation*}
  we call the diagram
  \begin{equation*}
    P_1 \stackrel{J_1}{\longleftarrow} S \stackrel{J_2}{\longrightarrow} P_2
  \end{equation*}
  a \emph{weak dual pair}.
  \end{defn}
  We would have a proper dual pair if the kernel inclusions were replaced by equalities.

  \begin{thm} \label{thm:dual_pairs}
    Let $J_1,J_2:S \to P_1,P_2$ form a weak dual pair.
    Let $h \in C^1(P_1)$.
    Let $x(t) \in S$ be a solution to Hamilton's equations
    with respect to the Hamiltonian $H = h \circ J_1$.
    Then $J_1\left( x(t) \right) \in P_1$ is a solution
    to Hamilton's equations on $P_1$ with respect to $h$,
    and $J_2( x(t))$ is constant in time.
  \end{thm}
  \begin{proof}
    Use proposition \ref{prop:Poisson_dynamics} to show that
    $\mu(t) = J_1 (x(t))$ is a solution to Hamilton's equations
    with respect to $h$.
    To verify that $J_2( x(t) )$ is constant,
    let $v \in \kernel(J_1)$ be a vector over $x(0) \in S$.
    This means that $v$ is tangent to the level set of $J_1$ at $x(0)$.
    Moreover, $H = h \circ J_1$ is constant on such level sets.
    Thus we observe
    \begin{equation*}
      0 = \langle dH(x(0)) , v \rangle =
      \omega \left( \dot{x} , v \right).
    \end{equation*}
    Since $v$ was an arbitrary element of $\kernel(J_1)$ over $x(0)$
    we see that $\dot{x} \in \kernel(J_1)^\omega$.
    Since $J_1$ and $J_2$ form a weak dual pair, this implies $\dot{x}  \in \kernel(J_2)$.
    Thus we have found
    \begin{equation*}
      \frac{d}{dt} J_2( x(t)) = TJ_2 \cdot \dot{x}(t)  
      = 0.
    \end{equation*}
  \end{proof}

\section{Diagrammatic overview} \label{app:diagramatic}
We present here a diagrammatic representation of some of the spaces used in the present paper. We begin by recalling a number of general results that hold for finite-dimensional Lie groups, before we indicate their relevance to the developments in the main text.

\begin{figure}[h]
\begin{center}
\begin{overpic}[scale = 0.7]{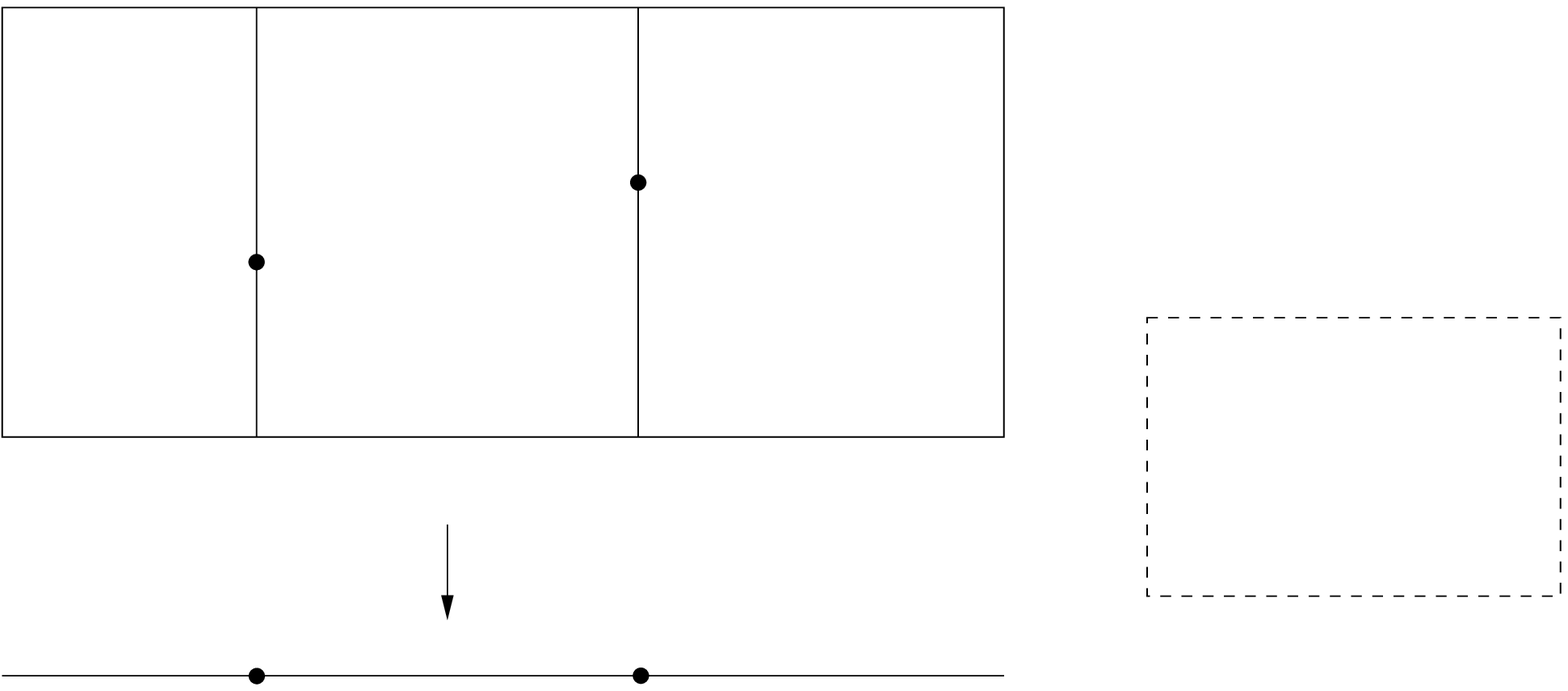}
\put(4, 6){\includegraphics[scale = 0.13]{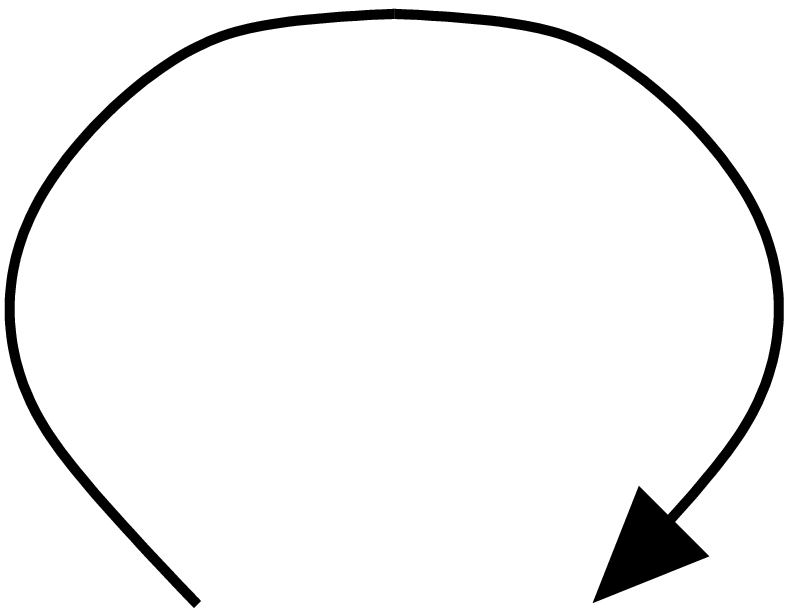}}
\put(6.5, 8){$\phi$}
\put(65, 2.5){\includegraphics[scale = 0.13]{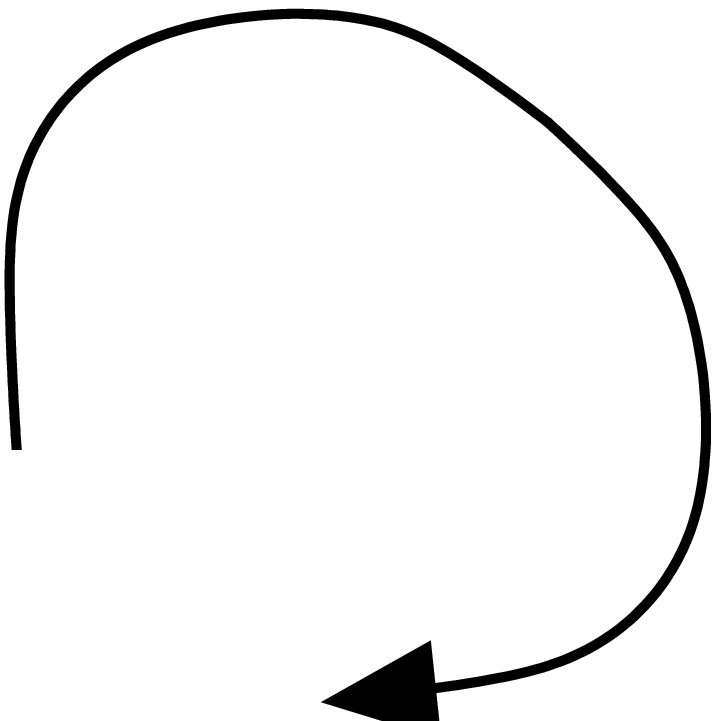}}
\put(67, 5){$\psi$}
\put(61, 42){$G$}
\put(42, 33) {$g$}
\put(18, 27.9) {$e$}
\put(17, 40) {$\updownarrow G_{q_0}$ }
\put(17, 4) {$q_0$ }
\put(35, 4) {$q = \Pi(q) = \phi_g(q_0)$ }
\put(25, 8) {$\Pi$ }
\put(-3, 1) {$Q$ }
\put(96.5, 22) {$K$ }
\end{overpic}
\end{center}
\caption{
\label{fig_diagram_1}
The $G$ action on $Q$ induces a projection $\Pi$. The group $K$ also acts on $G$, and we assume that the two actions commute.
}
\end{figure}

\begin{itemize}
\item
Let a Lie group $G$ act on a manifold $Q$ by the action $\phi: G \times Q \to Q$, which we also write as $\phi_g(\cdot) = \phi( g, \cdot )$. If we fix a particular value $q_0 \in Q$, we can construct a mapping $\Pi: G \to Q$ given by $\Pi(g) = \phi_g(q_0)$. Let us assume that the action is transitive, so that $\Pi$ is surjective. We denote by $G_{q_0}$ the isotropy subgroup leaving $q_0$ invariant, that is,
\begin{equation*}
  G_{q_0} := \{ g \in G | \phi_g(q_0) = q_0\}.
\end{equation*}
Note that $\Pi^{-1}(q) =g  G_{q_0}$ 
for any $g \in G$ such that $\Pi(g) = q$, and hence we can identify $Q$  with $G/G_{q_0}$. Suppose a further Lie group, $K$, also acts on $Q$ with group action $\psi: Q \times K \to Q$, which commutes with $\phi$. This situation arises naturally, for instance, when $K$ is a subgroup of $G$ and, in turn,  $G_{q_0}$ is a normal subgroup of $K$. In that case, one can define the action $\psi$ as
\begin{equation}
  \psi_k(q) := \Pi( g k) = \phi_{gk}(q_0), \label{right_ac_general}
\end{equation}
and check that $\phi$ and $\psi$ indeed commute:
\begin{equation*}
  \phi_s(\psi_k (q)) = \phi_s(\phi_{gk}(q_0)) = \phi_{sgk}(q_0) = \psi_k(\phi_s (q)).
\end{equation*}
We refer to Figure \ref{fig_diagram_1} for a representation of the relevant spaces and maps.
\item The actions $\phi$ and $\psi$ on $Q$ can be lifted to actions $\Phi$ and $\Psi$ on the cotangent bundle $T^*Q$ in the usual manner (see Figure \ref{fig_diagram_2}). These cotangent lifted actions induce equivariant momentum maps $J_1: T^*Q \to \mathfrak{g}^*$ and $J_2:T^*Q \to \mathfrak{k}^*$, where $\mathfrak{g}^*$ and $\mathfrak{k}^*$ are the duals of the Lie algebras of $G$ and $K$. Due to their equivariance, $J_1$ and $J_2$ are Poisson maps (where the duals of the Lie algebras are equipped with appropriate Lie--Poisson brackets). Since the actions $\phi$ and $\psi$ commute, the action $\Psi$ leaves level sets of $J_1$ invariant, and vice versa. This implies that $J_1$ and $J_2$ are a weak dual pair, and if moreover $\Psi$ is transitive on the level sets of $J_1$ and vice versa, then $J_1$ and $J_2$ are a proper dual pair, see~\cite[Corollary~2.6]{GayBalmazVizman2012}.
\item Let $\mathcal{H}:T^*G \to \mathbb{R}$ be a right-invariant Hamiltonian. This means that $\mathcal{H}$ is invariant with respect to the  cotangent lift $TR^*$ of the multiplication from the right of $G$ by itself. In particular, the reduced Hamiltonian $h:\mathfrak{g}^* \to \mathbb{R}$ satisfies $\mathcal{H}(\alpha_g) = h \circ TR_g^*(\alpha_g)$ for any $\alpha_g \in T^*G$, and the reduced dynamics in $\mathfrak{g}^*$ are of Lie--Poisson type. The momentum map $J_1$ can be used to induce the so-called \emph{collective Hamiltonian} $H = h\circ J_1$ on $T^*Q$. Note that $T^*Q$ is a symplectic manifold, and that the symplectic (canonical) dynamics with respect to the collective Hamiltonian are mapped by (the Poisson map) $J_1$ to the reduced dynamics on $\mathfrak{g}^*$. Moreover, $J_2$ is conserved under the dynamics on $T^*Q$. The conservation law follows from Noether's theorem because $J_1$, and hence $H$, are left invariant by $\Psi$ (see Theorem \ref{thm:dual_pairs}).  Note that the elements of $T^*Q$ play the role of symplectic variables (or Clebsch variables in the sense of \cite{MarsdenWeinstein1983}).
\item The appeal of Clebsch variables is their symplectic nature. The symmetry of $H$ with respect to $\Psi$ implies that reduced dynamics on $T^*G/H$ can be constructed by symplectic reduction. Note however that the resulting quotient manifold is not symplectic in general.
\item
Note also that there is a symplectic diffeomorphism between $T^*Q$ and $J^{-1}(0)/G_{q_0}$, where $J$ here is the momentum map associated with the cotangent lift of the action (from the right) of $G_{q_0}$ on $G$, see \cite[Theorem~2.2.2]{HRS}.
\end{itemize}
\begin{figure}[h]
\begin{center}
\begin{overpic}[scale = 0.6]{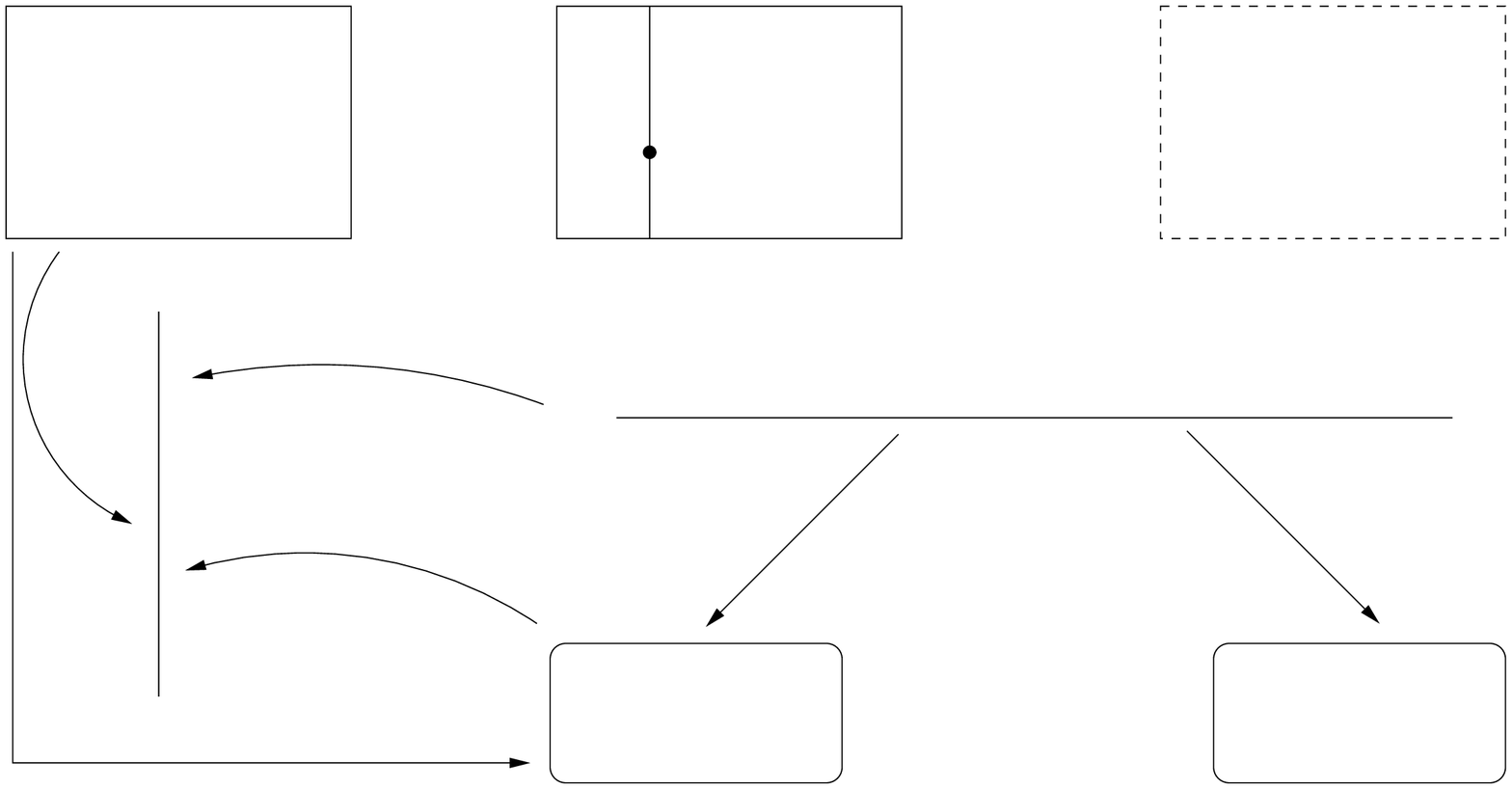}
\put(45, 28){\includegraphics[scale = 0.13]{arrow.eps}}
\put(47.8, 30){$\Phi$}
\put(84, 28){\includegraphics[scale = 0.13]{arrow.eps}}
\put(86.8, 30){$\Psi$}
\put(57, 49){$G$}
\put(44, 41.2) {$e$}
\put(43.3, 47) {$\updownarrow G_{q_0}$ }
\put(97, 23.5) {$T^*Q$ }
\put(96.5, 49) {$K$ }
\put(27, 40){\includegraphics[scale = 0.13]{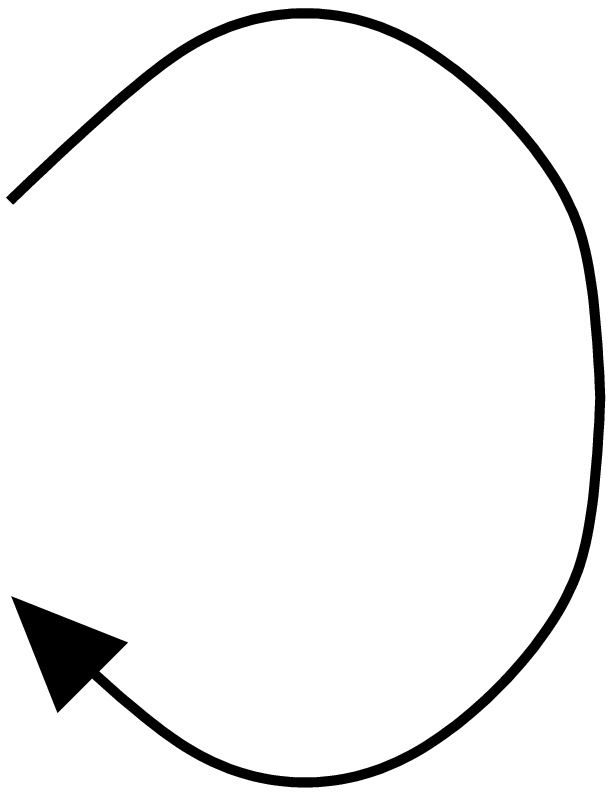}}
\put(17, 49) {$T^*G$ }
\put(3, 25) {$\mathcal{H}$ }
\put(16, 24.5) {$H = h \circ J_1$ }
\put(22, 12) {$h $ }
\put(7, 7) {$\mathbb{R} $ }
\put(52.5, 7) {$\mathfrak{g}^*$ }
\put(96.5, 7) {$\mathfrak{k}^*$ }
\put(49.5, 17) {$J_1$ }
\put(85.5, 17) {$J_2$ }
\put(22, 2.5) {$TR^* $ }
\end{overpic}
\end{center}
\caption{The cotangent lifted actions induce momentum maps $J_1$ and $J_2$. The Hamiltonian $\mathcal{H}$, assumed to be right-invariant, is also shown, along with the reduced Hamiltonian $H$ and the collective Hamiltonian $H = h \circ J_1$.
\label{fig_diagram_2}
 }
\end{figure}

In translating the above facts to the case of interest in the present paper, one encounters technical subtleties to do with the infinite-dimensionality of $\SDiff(\R^n)$. Nevertheless, bullet-by-bullet parallels can be recognized between the developments in the main text of the paper and the general results above, as we will discuss now. For simplicity, we restrict ourselves in what follows to the case of a single particle, the extension to $N$ particles being straightforward.
\begin{itemize}
\item Let $G = \SDiff(\R^n)$, and let $Q= Q_1^{(k)}$ be the space of single-particle $k$-jetlets. We fix the point $q_0 = (z, \mathbf{1}, 0, \ldots, 0) \in Q_1^{(k)}$ corresponding to the Taylor expansion of the identity in $\SDiff(\R^n)$ evaluated at $z \in \R^n$ (cf. Section \ref{sec:higher_order}). We let the projection $\Pi$ be given by ${\rm Jet}_z^{(k)}$, and hence the left action of an element $\varphi \in \SDiff(\R^n)$ on $q \in Q$ is $\phi_\varphi(q) = {\rm Jet}_z^{(k)}(\varphi \circ \rho)$, for any $\rho$ such that ${\rm Jet}_z^{(k)}(\rho) = q$. The role of the isotropy subgroup $G_{q_0}$ is played by ${\rm iso}^{(k)}_z$. Let $K = {\rm iso}(z)$, and note that ${\rm iso}^{(k)}(z)$ is a normal subgroup of ${\rm iso}(z)$ (\cite[Proposition 4.1]{JacobsRatiuDesbrun2013}). Hence, we can define a right action of ${\rm iso}(z)$ on $Q_1^{(k)}$ given by \eqref{right_ac_general}, namely
  \begin{equation*}
    \psi_\rho(q) = {\rm Jet}_z^{(k)}(\varphi \circ \rho),
  \end{equation*}
for any $\varphi$ such that ${\rm Jet}_z^{(k)}(\varphi) = q$.
\item The cotangent lift of the left and right actions on $Q$ lead to the momentum maps $J_L^{(k)}$ and $J_R^{(k)}$, explicitly computed in Section  \ref{sec:higher_order} and shown to be a weak dual pair in Proposition \ref{prop:k dual pairs}.
\item  In Theorem \ref{thm:k-solutions} we showed that the canonical Hamiltonian equations on $T^*{Q_1}^{(k)}$ associated with the collective Hamiltonian $H^{(k)} = h_{p, \sigma} \circ J_L^{(k)}$ lead to trajectories that are mapped, by $J_L^{(k)}$, to solutions of Hamilton's equations on $\mathfrak{X}_{\rm div}(\R^n)^*$. Moreover, we showed that $J_R^{(k)}$ is a constant of motion.
\item We briefly visit further symplectic reduction in Appendix \ref{sec:eom}. For more on this topic we refer to \cite{JacobsRatiuDesbrun2013, CotterHolmJacobsMeier2014}.
\item In Section \ref{sec:Kelvin} we constructed the mapping $i:T^*Q_1^{(k)} \to T^*\SDiff(\R^n)/{\rm iso}^{(k)}_z$ given by \eqref{i_explicit}, namely,
\begin{equation*}
 	i(q, p) = [TR^*_{\varphi^{-1}} J_L^{(k)}(q, p)].
 \end{equation*}
As we briefly mentioned towards the end of that section, if we denote by $J_{\rm conv}^{(k)}: T^*\SDiff(\R^n) \to \mathfrak{iso}^{(k)}(z)^*$ the momentum map associated to the cotangent lift of the right action of $\iso^{(k)}(z)$ on $\SDiff(\R^n)$,
then one can show that $i\bigl(T^*Q^{(k)}_1\bigr) = (J_{\rm conv}^{(k)})^{-1}(0) / \iso^{(k)}(z)$, as suggested by \cite[Theorem~2.2.2]{HRS}.
\end{itemize}

\section{Multi-indices}
\label{sec:multi}
A multiset is a set with some notion of multiplicity \cite{Blizard1989}.
In this paper, a multi-index on $\R^n$ is a multiset of elements derived from the generating set $\{1,\dots,n\}$.
Heuristically, a multi-index is just a ``bag of marbles'' each of which comes in $n$ ``colors''.
Given two multi-indices $\alpha \in \bag^i(n)$ and $\beta \in \bag^j(n)$ one can create the
multiset union $\alpha \cup \beta \in \bag^{i+j}(n)$ by collecting the marbles of $\alpha$ and $\beta$ into a single bag.
Given integers $b_1,\dots,b_j \in \R^n$, we can define the unique multi-index $\beta = [b_1,\dots,b_j] \in \bag^j(n)$
obtained by collecting $b_1,\dots,b_j$ into a bag.
Given these conventions, we can denote the partial differential operator $\partial_{b_1 \cdots b_n }$ by $\partial_\beta$.
Moreover, the notion of equivalence of mixed partials is expressed by the equivalence 
$\partial_\alpha \partial_\beta = \partial_{\alpha \cup \beta} = \partial_\beta \partial_\alpha$.
The cardinality of the multi-index $\beta$ is denoted $|\beta|$ and is given by the number of marbles in the bag.
Thus the order of the partial differential operator $\partial_{\beta}$ is $|\beta|$.
We denote the space of $k$-th order partitions of a multi-index by $\Pi(\alpha,k)$.
Rather than defining all this formally, we will compute an example and refer to \cite{Jacobs2014b} for the formal definitions.

We can consider the integers $1$ $2$ and $1$, and the partial differential operator $\partial_{121}$.
The associated multi-index is just $[1,2,1]$.  This multi-index is equivalent to the multi-index $[1,1,2]$ and $[2,1,1]$.
We say that it contains the elements $1$ and $2$.
Because it contains `$1$' two times, we say that the multiplicity of $1$ is $2$.
The multiset of $2$-fold partitions is $\Pi( [1,2,1],2)$, and consists of three multiset-partitions
\begin{equation*}
	[ [1] , [2,1] ] , [[1,1] , [2] ] , [ [1, 2] , [1] ].
\end{equation*}
Note that the first and the third partition correspond to the same multiset.
The cardinality of $\Pi( [1,2,1],2)$ is $3$, although it only has two distinct elements (one with a multiplicity of $1$, and another with a multiplicity of $2$).

\section{The dual space to divergence free vector fields}
\label{sec:measure_valued_momap}
In this section we will provide a terse and incomplete
characterization of the dual space of divergence free
vector fields.  First let us characterize the dual space of the space of all vector fields (with ``proper'' decay).
Let $\mathfrak{X}(\R^n)$ be the space of vector fields which decay at infinity
in such a way that $\mathfrak{X}(\R^n)$ is Fr\'echet.
Viewing $\mathfrak{X}(\R^n)$ as a subspace of functions from $\R^n$ to $\R^n$
we can view its dual as a space of distributions.
That is to say, given any $m \in \mathfrak{X}(\R^n)^*$ we may write $m$ as a tensor
product $p \otimes \mu$ where $\mu$ is a distribution (perhaps a measure)
on $\R^n$ and $p$ is a covector field (i.e.\ a one-form).
Conversely, given any $p \in \Omega^1(\R^n)$ and distribution $\mu \in \mathcal{D}(\R^n)$
we may form the tensor product $p \otimes \mu$.
The object $p\otimes \mu$ is identified as an element of $\mathfrak{X}(\R^n)^*$ through
the pairing
\begin{equation*}
	\langle p \otimes \mu , u \rangle :=  \int_{\R^n} \langle p(x) , u(x) \rangle \mu.
\end{equation*}
where $\langle p(x),u(x) \rangle$ is the function on $\R^n$ obtained by pairing the covector $p(x) \in T_x^*\R^n$ with the vector $u(x) \in T_x\R^n$.
If we restrict ourselves to the case of divergence free vector fields, we need to quotient
the dual space appropriately.  In particular, we see that the annihilator of $\mathfrak{X}_{\rm div}(\R^n)$
as a subspace of $\mathfrak{X}(\R^n)^*$ is
\begin{align*}
	(\mathfrak{X}_{\rm div}(\R^n) )^\circ := \{ m \in \mathfrak{X}(\R^n)^* \mid \langle m , u \rangle = 0 , \forall u \in \mathfrak{X}_{\rm div}(\R^n) \}\\
		:= {\rm closure}\{ p \otimes \dx \in \mathfrak{X}(\R^n)^* \mid p \text{ is a closed one-form} \}.
\end{align*}
where $\dx$ is the canonical volume form on $\R^n$ and
we have used the fact that the gradient fields and the harmonic vector fields are $L^2$-orthogonal to the divergence free vector fields.
The dual space $\mathfrak{X}_{\rm div}(\R^n)^*$ is identical to the quotient space $\mathfrak{X}(\R^n)^* / (\mathfrak{X}_{\rm div}(\R^n))^\circ$.
In other words, we may view a $m \in \mathfrak{X}_{\rm div}(\R^n)^*$ as an object of the form $p \otimes \mu$ modulo $\mathfrak{X}_{\rm div}(\R^n)^\circ$.
In the text we will typically not mention $\mathfrak{X}_{\rm div}(\R^n)^\circ$ explicitly, and simply identify $m$ with $p \otimes \mu$.
This is a harmless identification as long as we do not pair it with a non-divergence free vector field.

\section{Equations of motion for 1-jetlets}
\label{sec:eom}

The equations of motion are expressible as Hamiltonian
equations on $T^*Q_N^{(1)}$
in canonical variable $(q^{(0)},p^{(0)},q^{(1)},p^{(1)})$.
However, it is more efficient to express the equations of 
motion in the non-canonical variables $(q,p,\mu)$
where $q_a = q_a^{(0)}$, $p_a = p_a^{(0)}$ and $\mu_a = [q_a^{(1)}]\indices{^\ell_i} [p_a^{(1)}] \indices{_\ell^j}$ for $a = 1,\dots,N$.
The Hamiltonian in these coordinates is
\begin{align*}
  H(q,p,\mu) &=
  \frac{1}{2} p_{a i} p_{b j} K^{i j}(q_a - q_b) - p_{a i}
  [\mu_b]\indices{_{j}^{k}} \partial_k K^{i j}(q_a - q_b) \\
& -\frac{1}{2} [\mu_a^{(1)}]\indices{_{i}^l}
  [\mu_b^{(1)}]\indices{_j^k} \partial_{l k} K^{i j}(q_a - q_b),
\end{align*}
where $K^{i j}(x) = \delta^{i j} e^{- \| x \|^2 / 2 \sigma^2}$.
Hamilton's equations are then given in short by
\begin{align}
  \dot{q} &=  \frac{\partial H}{\partial p} \label{eq:dHdp} \\
  \dot{p} &= -\frac{\partial H}{\partial q} \label{eq:dHdq} \\
  \xi     &=  \frac{\partial H}{\partial \mu} \label{eq:dHdmu} \\
  \dot{\mu} &= - \ad^*_{ \xi } ( \mu ) \label{eq:LiePoisson}
\end{align}
where $\ad^*$ refers to the coadjoint operator on $\SL(n)$.
More explicitly, equation~\eqref{eq:dHdp} is given by
\begin{equation*}
  \dot{q}_a^i = p_{b j} K^{i j}(q_a - q_b)
 -[\mu^{(1)}_b] \indices{_j^k} \partial_k K^{i j}(q_a - q_b)
\end{equation*}
equation~\eqref{eq:dHdq} is given by the sum
\begin{equation*}
  \dot{p}_{a i} = T^{00}_{a i} + T^{01}_{a i} + T^{11}_{a i},
\end{equation*}
where we define the three terms in this sum as
\begin{align*}
  T^{00}_{a i} &= -p_{a k} p_{b j} \partial_{i} K^{k j}(q_a - q_b) \\
  T^{01}_{a i} &= (p_{a l}[\mu_b^{(1)}]\indices{_j^k}
                  -p_{b l}[\mu_a^{(1)}]\indices{_j^k}) \partial_{k i} K^{l j}(q_a - q_b) \\
  T^{11}_{a i} &= [\mu_a^{(1)}] \indices{_m^l} [\mu_b^{(1)}]\indices{_j^k}
                  \partial_{l k i} K^{m j}(q_a - q_b). \\
\end{align*}
Next, we calculate the quantities $\xi = \partial H / \partial \mu$
for $k=1,2$ of equation~\eqref{eq:dHdmu} to be
\begin{equation*}
  [\xi_a ]\indices{^i_j} =
  p_{b k} \partial_j K^{i k}(q_a - q_b)
 -[\mu_b]\indices{_{l}^{k}} \partial_{j k} K^{i l}(q_a - q_b),
\end{equation*}
which allows us to compute $\dot{\mu}$ in equation~\eqref{eq:LiePoisson} as
\begin{equation*}
  [\dot{\mu}_a ]\indices{_i^j} =
  [\mu_a]\indices{_i^k} [\xi_a^{(1)}]\indices{^j_k}
 -[\mu_a]\indices{_k^j} [\xi_a^{(1)}]\indices{^k_i}.
\end{equation*}

The dynamics in terms of the original variables $(q^{(0)},q^{(1)})$ with $q^{(1)} \in \SL(n)$ are obtained by integrating the reconstruction equations
$
  [\dot{q}^{(1)}]\indices{^i_j} =
  [\xi]\indices{^i_k} [q^{(1)}]\indices{^k_j}
$.

\bibliographystyle{amsalpha}
\bibliography{hoj_2014}

\end{document}